\documentclass[reqno]{amsart}

\usepackage{amsaddr}
\usepackage{amsmath}
\usepackage{amssymb}
\usepackage{amsthm}
\usepackage{eucal}
\usepackage{mathrsfs}
\usepackage{fullpage}
\usepackage{setspace}
\usepackage{bm}
\usepackage{color}
\usepackage{tabularx}
\usepackage{mathabx}
\usepackage{graphicx}
\usepackage[numbers,square,sort&compress]{natbib}
\usepackage[all,cmtip]{xy}
\usepackage{lineno}
\usepackage{subfig}
\usepackage{floatrow}

\title{autocratic strategies for alternating games}
\author{alex mcavoy and christoph hauert}

\theoremstyle{definition}

\newtheorem{lemma}{Lemma}
\newtheorem{proposition}{Proposition}

\newtheorem*{unlabeledRemark}{Remark}
\newtheorem{theorem}{Theorem}

\newtheorem*{strictXTheorem}{Theorem \ref{thm:strictX}}
\newtheorem*{strictYTheorem}{Theorem \ref{thm:strictY}}
\newtheorem*{randomTheorem}{Theorem \ref{thm:random}}

\begin{document}

\begin{abstract}
Repeated games have a long tradition in the behavioral sciences and evolutionary biology. Recently, strategies were discovered that permit an unprecedented level of control over repeated interactions by enabling a player to unilaterally enforce linear constraints on payoffs. Here, we extend this theory of ``zero-determinant" (or, more generally, ``autocratic") strategies to alternating games, which are often biologically more relevant than traditional synchronous games. Alternating games naturally result in asymmetries between players because the first move matters or because players might not move with equal probabilities. In a strictly-alternating game with two players, $X$ and $Y$, we give conditions for the existence of autocratic strategies for player $X$ when (i) $X$ moves first and (ii) $Y$ moves first. Furthermore, we show that autocratic strategies exist even for (iii) games with randomly-alternating moves. Particularly important categories of autocratic strategies are extortionate and generous strategies, which enforce unfavorable and favorable outcomes for the opponent, respectively. We illustrate these strategies using the continuous Donation Game, in which a player pays a cost to provide a benefit to the opponent according to a continuous cooperative investment level. Asymmetries due to alternating moves could easily arise from dominance hierarchies, and we show that they can endow subordinate players with more autocratic strategies than dominant players.
\end{abstract}

\maketitle

\section{Introduction}

Repeated games, and, in particular, the repeated Prisoner's Dilemma, have been used extensively to study the reciprocation of cooperative behaviors in social dilemmas \citep{trivers:TQRB:1971,axelrod:Science:1981,axelrod:BB:1984,nowak:Science:2006}. These games traditionally involve a sequence of interactions in which two players act simultaneously (or, at least without knowing the opponent's move) and condition their decisions on the history of their previous encounters. Even though such synchronized decisions seem often contrived in realistic social interactions, the biologically more realistic and relevant scenario with asynchronous interactions has received surprisingly little attention. In asynchronous games, players take turns and alternate moves in either a strict or random fashion \citep{nowak:JTB:1994,wedekind:PNAS:1996}.

A classic example of an asynchronous game with alternating moves is blood donation in vampire bats \citep{wilkinson:Nature:1984}. When a well-fed bat donates blood to a hungry fellow, the recipient has the opportunity to return the favor at a later time. Similarly, social grooming between two primates is not always performed simultaneously; instead, one animal grooms another, who then has the opportunity to reciprocate in the future \citep{muroyama:B:1991}. On a smaller scale, the biosynthesis of iron-scavenging compounds by microorganisms through quorum sensing can result in asynchronous responses to fellow ``players" in the population \citep{stintzi:FEMSML:1998,miller:ARM:2001,iliopoulos:PLoSCB:2010}. Even for interactions that appear to involve simultaneous decisions, such as in acts of predator inspection by fish \citep{milinski:Nature:1987}, it remains difficult to rule out that these interactions are not instead based on rapid, non-synchronous decisions \citep{frean:PRSB:1994}.

The iterated Prisoner's Dilemma game, which involves a choice to either cooperate, $C$, or defect, $D$, in each round, has played a central role in the study of reciprocal altruism \citep{axelrod:Science:1981,axelrod:BB:1984,nowak:Science:2006}. Rather unexpectedly, after decades of intense study of iterated games, \citet{press:PNAS:2012} showed that a player can unilaterally enforce linear payoff relationships in synchronous games. For example, if $\pi_{X}$ and $\pi_{Y}$ are the expected payoffs to players $X$ and $Y$, respectively, and $\chi\geqslant 1$ is an extortion factor, then player $X$ can ensure that $\pi_{X}=\chi\pi_{Y}$, regardless of the strategy of player $Y$. Moreover, such linear relationships may be enforced using merely memory-one strategies, which condition the next move on the outcome of just the previous round.

The discovery of these so-called ``zero-determinant" strategies triggered a flurry of follow-up studies. Most notably, from an evolutionary perspective, extortionate strategies fare poorly \citep{hilbe:PNAS:2013} but can be stable provided that extortioners recognize one another \citep{adami:NC:2013}. However, generous counterparts of extortionate strategies perform much better in evolving populations \citep{stewart:PNAS:2012,stewart:PNAS:2013} and constitute Nash equilibria for the repeated Prisoner's Dilemma \citep{hilbe:GEB:2015} (but generally only if there are just two discrete levels of cooperation \citep{mcavoy:PNAS:2016}). Against humans, extortionate strategies typically underperform generous strategies when the extortioner is also a human \citep{hilbe:NC:2014} but can outperform generous strategies when the extortioner is a computer \citep{wang:NC:2016}. Thus, for the settings in which zero-determinant strategies are known to exist, their performance is sensitive to the context in which they arise. Our focus here is on extending these strategies further into the domain of alternating interactions from a classical, non-evolutionary viewpoint. In particular, we establish the existence of zero-determinant strategies for several types of alternating interactions between two players.

Recently, autocratic strategies were introduced as a generalization of zero-determinant strategies to simultaneous games with arbitrary action spaces \citep{mcavoy:PNAS:2016}. An autocratic strategy for player $X$ is any strategy that, for some constants $\alpha$, $\beta$, and $\gamma$ (not all zero), enforces the linear relationship
\begin{linenomath}
\begin{align}
\label{eq:autocratic}
\alpha\pi_{X} + \beta\pi_{Y} + \gamma &= 0
\end{align}
\end{linenomath}
on expected payoffs every strategy of player $Y$. Here, we consider autocratic strategies in alternating games. In a strictly-alternating game, one player moves first (either $X$ or $Y$) and waits for the opponent's response before moving again. This process then repeats, with each player moving only after the opponent moved (see Fig. \ref{fig:strictIStrictIIRandom}(A),(B)). In contrast, in a randomly-alternating game, the player who moves in each round is chosen stochastically: at each time step, $X$ moves with probability $\omega_{X}$ and $Y$ moves with probability $1-\omega_{X}$ for some $0\leqslant\omega_{X}\leqslant 1$ (see Fig. \ref{fig:strictIStrictIIRandom}(C)). Note that only for $\omega_X=1/2$ is it the case that both players move, on average, equally often.

\newcolumntype{C}{>{\centering\arraybackslash}p{3em}}
\floatsetup[figure]{style=plain,subcapbesideposition=top}
\begin{figure}[tp]
\centering
\sidesubfloat[]{
\begin{tabular}{ c | C C C C C C C }
\textrm{player }$X$ & ${\color{blue}C}$ &  & ${\color{red}D}$ & & ${\color{red}D}$ & & \\
\textrm{player }$Y$ &  & ${\color{red}D}$ &  & ${\color{blue}C}$ & & ${\color{blue}C}$ & \\
\hline
\textrm{round} & $0$ & $1$ & $2$ & $3$ & $4$ & $5$ & $\cdots$ \\
 & $\downarrow$ & $\downarrow$ & $\downarrow$ & $\downarrow$ & $\downarrow$ & $\downarrow$ & \\
$\left(X\textrm{'s payoff},Y\textrm{'s payoff}\right)$ & $\left(-c,b\right)$ & $\left(0,0\right)$ & $\left(0,0\right)$ & $\left(b,-c\right)$ & $\left(0,0\right)$ & $\left(b,-c\right)$ &
\end{tabular}
}%
\newline\newline\newline
\sidesubfloat[]{
\begin{tabular}{ c | C C C C C C C }
\textrm{player }$X$ &  & ${\color{blue}C}$ &  & ${\color{red}D}$ & & ${\color{red}D}$ & \\
\textrm{player }$Y$ & ${\color{red}D}$ &  & ${\color{blue}C}$ & & ${\color{blue}C}$ & & \\
\hline
\textrm{round} & $0$ & $1$ & $2$ & $3$ & $4$ & $5$ & $\cdots$ \\
  & $\downarrow$ & $\downarrow$ & $\downarrow$ & $\downarrow$ & $\downarrow$ & $\downarrow$ & \\
$\left(X\textrm{'s payoff},Y\textrm{'s payoff}\right)$ & $\left(0,0\right)$ & $\left(-c,b\right)$ & $\left(b,-c\right)$ & $\left(0,0\right)$ & $\left(b,-c\right)$ & $\left(0,0\right)$ &
\end{tabular}
}%
\newline\newline\newline
\sidesubfloat[]{
\begin{tabular}{ c | C C C C C C C }
\textrm{player }$X$ & ${\color{blue}C}$ &  & ${\color{red}D}$ & ${\color{red}D}$ & ${\color{blue}C}$ & & \\
\textrm{player }$Y$ &  & ${\color{red}D}$ &  & & & ${\color{blue}C}$ & \\
\hline
\textrm{round} & $0$ & $1$ & $2$ & $3$ & $4$ & $5$ & $\cdots$ \\
  & $\downarrow$ & $\downarrow$ & $\downarrow$ & $\downarrow$ & $\downarrow$ & $\downarrow$ & \\
$\left(X\textrm{'s payoff},Y\textrm{'s payoff}\right)$ & $\left(-c,b\right)$ & $\left(0,0\right)$ & $\left(0,0\right)$ & $\left(0,0\right)$ & $\left(-c,b\right)$ & $\left(b,-c\right)$ &
\end{tabular}
}%
\newline
\caption{Three types of interactions in the alternating Donation Game: (A) strictly-alternating game in which player $X$ moves first; (B) strictly-alternating game in which player $Y$ moves first; and (C) randomly-alternating game in which, in each round, player $X$ moves with probability $\omega_{X}$ and player $Y$ with probability $1-\omega_{X}$. For each type of alternating game, a player moves either $C$ or $D$ (cooperate or defect) in each round and both players receive a payoff from this move. Unlike in strictly-alternating games, (A) and (B), a player might move several times in a row in a randomly-alternating game, (C).\label{fig:strictIStrictIIRandom}}
\end{figure}

Previous studies of zero-determinant strategies have focused on enforcing linear payoff relationships using conditional responses with short memories. A player using a memory-one strategy determines his or her response (stochastically) based on the outcome of just the previous round. Although strategies with longer memory length have been shown to help establish cooperation \citep{hauert:PRSLB:1997,stewart:SR:2016}, they are not always reliably implemented in players with limited memory capacity (including humans) \citep{milinski:PNAS:1998,stevens:FP:2011,baek:SR:2016}. Here, we follow the tradition of concentrating on shorter-memory strategies. In particular, we establish the existence of memory-one autocratic strategies for alternating games and give several simple examples that enforce linear payoff relationships for every strategy of the opponent (even those with unlimited memory).

In the classical Donation Game \citep{sigmund:PUP:2010}, a player either (i) cooperates and donates $b$ to the opponent at a cost of $c$ or (ii) defects and donates nothing and pays no cost, which yields the payoff matrix
\begin{linenomath}
\begin{align}\label{eq:donationPayoffMatrix}
\bordermatrix{%
 & C & D \cr
C &\ b-c & \ -c \cr
D &\ b & \ 0 \cr
}
\end{align}
\end{linenomath}
and represents an instance of the Prisoner's Dilemma provided that benefits exceed the costs, $b>c>0$. The continuous Donation Game extends this binary action space and allows for a continuous range of cooperation levels \citep{killingback:PRSB:1999,wahl:JTB:1999a,wahl:JTB:1999b,killingback:AN:2002}. An action in this game is an investment level, $s$, taken from an interval, $\left[0,K\right]$, where $K$ indicates an upper bound on investments. Based on its investment level, $s$, a player then pays a cost of $c\left(s\right)$ to donate $b\left(s\right)$ to the opponent where $b(s)$ and $c(s)$ are continuous non-decreasing functions with $b(s)>c(s)>0$ for $s>0$ and $b(0)=c(0)=0$; an investment of zero corresponds to defection, which neither generates benefits nor incurs costs \citep{killingback:AN:2002}. Biologically-relevant interpretations of continuous investment levels (as well as alternating moves) include (i) the effort expended in social grooming and ectoparasite removal by primates \citep{dunbar:FP:1991}; (ii) the quantity of blood donated by one vampire bat to another \citep{wilkinson:Nature:1984}; (iii) the amount of iron-binding agents (siderophores) produced by bacterial parasites \citep{west:PRSB:2003}; and (iv) the honesty level of a (human) party involved in a trade agreement \citep{verhoeff:CSN:1998}.

In alternating games, the assignment of payoffs to players deserves closer inspection \citep{hauert:JTB:1998}. Here, we focus on alternating games in which both players obtain payoffs after every move (like in the continuous Donation Game) \citep[see Fig.~\ref{fig:strictIStrictIIRandom};][]{nowak:JTB:1994}. Alternatively, payoffs could result from every pair of moves rather than every individual move \citep{frean:PRSB:1994}. While it is possible to construct a theory of autocratic strategies for strictly-alternating games in either setting, it becomes difficult to even define payoffs in the latter setup for randomly-alternating games because either player can move several times in a row (see Fig.~\ref{fig:strictIStrictIIRandom}(C)). Therefore, we follow \citet{nowak:JTB:1994} in order to include the particularly relevant and intriguing case of randomly-alternating games. 

Randomly-alternating games seem more relevant for modeling biological interactions because often strict alternation cannot be maintained or enforced, or the players find themselves in different roles, which translate into different propensities to move. To accommodate these situations, we consider a class of randomly-alternating games in which the probability that player $X$ moves in a given round, $\omega_{X}$, is not necessarily $1/2$. Any other value of $\omega_{X}$ results in asymmetric interactions -- even if the payoffs in each encounter are symmetric -- simply because one player moves more often than the other. For example, dominance hierarchies in primates naturally result in asymmetric behavioral patterns \citep{mehlman:P:1988,lazaro_perea:AB:2004,newton_fisher:AB:2011}. In male chimpanzees, dominance hierarchies require smaller, subordinate chimpanzees to groom larger, dominant chimpanzees more often than vice versa \citep{foster:AJP:2009}. Therefore, including such asymmetries significantly expands the scope of interactions to which the theory of autocratic strategies applies.

\section{Existence of autocratic strategies}

In every round of an alternating game, either player $X$ or player $Y$ moves. On player $X$'s turn, she chooses an action, $x$, from an action space, $S_X$, and gets a payoff $f_{X}\left(x\right)$ while her opponent gets $f_{Y}\left(x\right)$. Similarly, when player $Y$ moves, he chooses an action, $y$, from $S_Y$ and gets a payoff $g_{Y}\left(y\right)$ while his opponent gets $g_{X}\left(y\right)$. Future payoffs are discounted by a factor $\lambda$ (with $0<\lambda <1$), which can represent a time preference \citep{fudenberg:MIT:1991} that is derived, for example, from interest rates for monetary payoffs. Alternatively, $\lambda$ can be interpreted as the probability that there will be another round, which results in a finitely-repeated game with an average length of $1/(1-\lambda)$ rounds \citep{nowak:Science:2006}.

\subsection{Strictly-alternating games}

In a pair of rounds in which player $X$ moves before $Y$, the payoffs are $u_{X}\left(x,y\right) :=f_{X}\left(x\right) +\lambda g_{X}\left(y\right)$ and $u_{Y}\left(x,y\right) :=f_{Y}\left(x\right) +\lambda g_{Y}\left(y\right)$, respectively. Note that the payoffs from $Y$'s move are discounted by a factor of $\lambda$ because $Y$ moves one round after $X$. The payoff functions, $u_{X}$ and $u_{Y}$, satisfy the ``equal gains from switching" property \citep{nowak:AAM:1990}, which means the difference between $u_{X}\left(x,y\right)$ and $u_{X}\left(x',y\right)$ is independent of the opponent's move, $y$. This property follows immediately from the fact that $u_{X}$ (or $u_{Y}$) is obtained by adding the separate contributions based on the moves of $X$ and $Y$.

Thus, if player $X$ moves first and $\left(x_{0},y_{1},x_{2},y_{3},\dots\right)$ is the sequence of play, then her average payoff is
\begin{linenomath}
\begin{align}\label{eq:expectedXFirst}
\pi_{X} &= \left(1-\lambda\right)\left[\sum_{t=0}^{\infty}\lambda^{2t}f_{X}\left(x_{2t}\right) + \sum_{t=0}^{\infty}\lambda^{2t+1}g_{X}\left(y_{2t+1}\right)\right] = \left(1-\lambda\right)\sum_{t=0}^{\infty}\lambda^{2t}u_{X}\left(x_{2t},y_{2t+1}\right) .
\end{align}
\end{linenomath}
The second expression resembles the average payoff for a simultaneous-move game whose one-shot payoff function is $u_{X}$ \citep{fudenberg:MIT:1991}. Similarly, replacing $u_{X}$ with $u_{Y}$ yields player $Y$'s average payoff, $\pi_{Y}$.

For strictly-alternating games, we borrow the term ``memory-one strategy" from synchronous games to mean a conditional response based on the previous moves of both players. Even though this memory now covers two rounds of interactions, it remains meaningful because player $Y$ always moves after player $X$ (or vice versa). For an arbitrary action space, $S_{X}$, a memory-one strategy for player $X$ formally consists of an initial action, $\sigma_{X}^{0}$, and a memory-one action, $\sigma_{X}\left[x,y\right]$, which are both probability distributions on $S_{X}$. Since player $X$ moves first, she bases her initial action on $\sigma_{X}^{0}$ and subsequently uses the two previous moves, $x$ and $y$, to choose an action in the next round using $\sigma_{X}\left[x,y\right]$ (see Fig. \ref{fig:memoryOneStrategies} for graphical depictions).

\begin{figure}
\begin{center}
\includegraphics[scale=0.45]{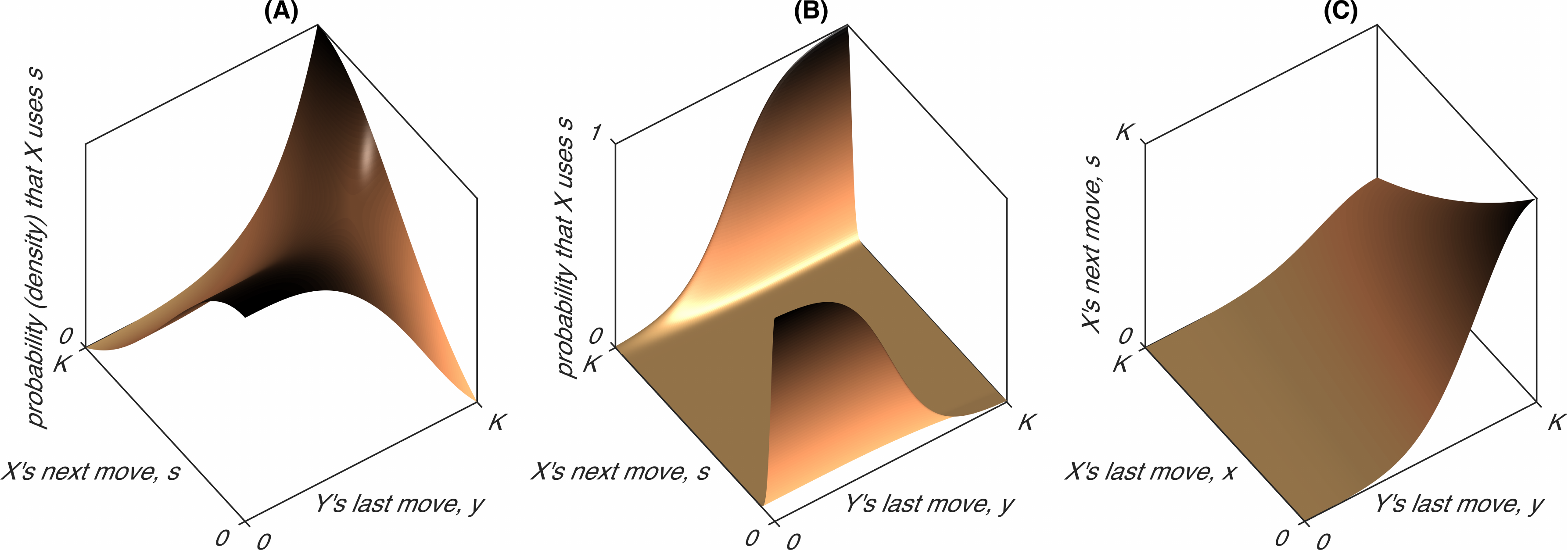}
\caption{Three examples of memory-one strategies for player $X$ in a strictly-alternating game whose action spaces are $S_{X}=S_{Y}=\left[0,K\right]$ for some $K>0$. (A) depicts a reactive stochastic strategy in which solely $Y$'s last move is used to determine the probability distribution with which $X$ chooses her next action. The mean of this distribution is an increasing function of $y$, which means that $X$ is more likely to invest more (play closer to $K$) as $y$ increases. (B) illustrates a reactive two-point strategy, i.e. a strategy that plays only two actions, $0$ (defect) or $K$ (fully cooperate). Player $Y$'s last move is used to determine the probability with which $X$ plays $K$ in the next round; if $X$ does not use $K$, then she uses $0$. As $Y$'s last action, $y$, increases, $X$ is more likely to reciprocate and use $K$ in response. (C) shows a strategy that gives $X$'s next move deterministically as a function of both of the players' last moves. Unlike in (A) and (B), $X$'s next move is uniquely determined by her own last move, $x$, and the last move of her opponent, $y$. If $Y$ used $y=0$ in the previous round, then $X$ responds by playing $0$ as well. $X$'s subsequent action is then an increasing function of $y$ whose rate of change is largest when $X$'s last move, $x$, is smallest. In particular, if $Y$ used $y>0$ in the previous round, then $X$'s next action is a decreasing function of her last move, $x$. Therefore, in (C), $X$ exploits players who are unconditional cooperators.
\label{fig:memoryOneStrategies}}
\end{center}
\end{figure}

\begin{theorem}[Autocratic strategies for strictly-alternating games in which $X$ moves first]\label{thm:strictX}
Suppose that
\begin{linenomath}
\begin{align}\label{eq:mainEquationStrictI}
\Big( \alpha &f_{X}\left(x\right) + \beta f_{Y}\left(x\right) + \gamma \Big) + \lambda\Big( \alpha g_{X}\left(y\right) + \beta g_{Y}\left(y\right) + \gamma \Big) \nonumber \\
&= \psi\left(x\right) - \lambda^{2}\int\limits_{s\in S_{X}}\psi\left(s\right)\,d\sigma_{X}\left[x,y\right]\left(s\right) - \left(1-\lambda^{2}\right)\int\limits_{s\in S_{X}}\psi\left(s\right)\,d\sigma_{X}^{0}\left(s\right)
\end{align}
\end{linenomath}
holds for some bounded $\psi$ and for each $x\in S_{X}$ and $y\in S_{Y}$. Then, if player $X$ moves first, the pair $\left(\sigma_{X}^{0},\sigma_{X}\left[x,y\right]\right)$ allows $X$ to enforce the equation $\alpha\pi_{X}+\beta\pi_{Y}+\gamma =0$ for every strategy of player $Y$.
\end{theorem}

A proof of this result may be found in Supporting Information.

Note that the average payoff, $\pi_{X}$, of Eq.~(\ref{eq:expectedXFirst}) is the same as in a simultaneous-move game whose payoff function is $\frac{1}{1+\lambda}u_{X}$ and whose discounting factor is $\lambda^{2}$ \citep{mcavoy:PNAS:2016}. Hence, it is not so surprising that autocratic strategies exist in this case too (and under similar conditions).
However, the situation changes if $Y$ moves first and $\left(y_{0},x_{1},y_{2},x_{3},\dots\right)$ is the sequence of play. In this case, $X$'s average payoff is
\begin{linenomath}
\begin{align}\label{eq:expectedYFirst}
\pi_{X} &= \left(1-\lambda\right)\left[g_{X}\left(y_{0}\right) + \lambda\sum_{t=0}^{\infty}\lambda^{2t}u_{X}\left(x_{2t+1},y_{2t+2}\right)\right] .
\end{align}
\end{linenomath}
When $Y$ moves first, player $X$'s initial move, $\sigma_{X}^{0}\left[y_{0}\right]$, is now a function of $Y$'s first move, $y_{0}$. However, $X$'s lack of control over the first round does not (in general) preclude the existence of autocratic strategies:
\begin{theorem}[Autocratic strategies for strictly-alternating games in which $Y$ moves first]\label{thm:strictY}
Suppose that
\begin{linenomath}
\begin{align}\label{eq:mainEquationStrictII}
\Big( \alpha &f_{X}\left(x\right) + \beta f_{Y}\left(x\right) + \gamma \Big) + \lambda\Big( \alpha g_{X}\left(y\right) + \beta g_{Y}\left(y\right) + \gamma \Big) + \left(\frac{1-\lambda^{2}}{\lambda}\right)\Big( \alpha g_{X}\left(y_{0}\right) + \beta g_{Y}\left(y_{0}\right) + \gamma \Big) \nonumber \\
&= \psi\left(x\right) - \lambda^{2}\int\limits_{s\in S_{X}}\psi\left(s\right)\,d\sigma_{X}\left[x,y\right]\left(s\right) - \left(1-\lambda^{2}\right)\int\limits_{s\in S_{X}}\psi\left(s\right)\,d\sigma_{X}^{0}\left[y_{0}\right]\left(s\right)
\end{align}
\end{linenomath}
holds for some bounded $\psi$ and for each $x\in S_{X}$ and $y_{0},y\in S_{Y}$. Then, if player $X$ moves second, the pair $\left(\sigma_{X}^{0}\left[y_{0}\right],\sigma_{X}\left[x,y\right]\right)$ allows $X$ to enforce the equation $\alpha\pi_{X}+\beta\pi_{Y}+\gamma =0$ for every strategy of player $Y$.
\end{theorem}

For a proof of this statement, see Supporting Information.

Note that Eq.~(\ref{eq:mainEquationStrictII}) is slightly more restrictive than Eq.~(\ref{eq:mainEquationStrictI}) because player $X$ has no control over the outcome of the initial round. Evidently, for undiscounted (infinite) games ($\lambda =1$), it is irrelevant who moves first and hence the conditions for the existence of autocratic strategies coincide (c.f. Eqs. (\ref{eq:mainEquationStrictI}) and (\ref{eq:mainEquationStrictII})).

\subsection{Randomly-alternating games}

In a randomly-alternating game, the player who moves in any given round is determined probabilistically: player $X$ moves with probability $\omega_{X}$ and player $Y$ with probability $1-\omega_{X}$. Suppose that $X$ and $Y$ each make plans to play $x_{t}$ and $y_{t}$ at time $t$, respectively, assuming they move at time $t$. Then, in the repeated game, these strategies give player $X$ an average payoff of
\begin{linenomath}
\begin{align}\label{eq:payRandom}
\pi_{X} &= \left(1-\lambda\right)\sum_{t=0}^{\infty}\lambda^{t}\Big( \omega_{X}f_{X}\left(x_{t}\right) + \left(1-\omega_{X}\right) g_{X}\left(y_{t}\right) \Big) .
\end{align}
\end{linenomath}
$Y$'s average payoff, $\pi_{Y}$, is obtained from Eq. (\ref{eq:payRandom}) by replacing $f_{X}$ and $g_{X}$ by $f_{Y}$ and $g_{Y}$, respectively.

For randomly-alternating games, we need to reconsider the concept of memory-one strategies. If moves alternate randomly, then a logical extension is provided by a conditional response based on the previous move as well as on which player moved. In particular, $\sigma_{X}^{Y}\left[y\right]$ denotes a mixed action for player $X$ after player $Y$ uses $y$ in the previous round, and $\sigma_{X}^{X}\left[x\right]$ denotes a mixed action after playing $x$ herself. Note that the cognitive requirement in terms of memory capacity in strictly-alternating games remains the same as for simultaneous games, whereas for randomly-alternating games the requirements are significantly less demanding as reflected in two univariate response functions as compared to response functions involving two variables. For two-action games (such as the classical Donation Game), however, memory-one strategies for synchronous, strictly-alternating, and randomly-alternating games all reduce to four-tuples of probabilities (see \S\ref{sec:classicalDG} below).

Rather surprisingly, randomly-alternating games also admit autocratic strategies:
\begin{theorem}[Autocratic strategies for randomly-alternating games]\label{thm:random}
If, for some bounded $\psi$,
\begin{linenomath}
\begin{subequations}\label{eq:mainEquationRandom}
\begin{align}
\alpha f_{X}\left(x\right) + \beta f_{Y}\left(x\right) + \gamma &= \psi\left(x\right) - \lambda \omega_{X}\int\limits_{s\in S_{X}}\psi\left(s\right)\,d\sigma_{X}^{X}\left[x\right]\left(s\right) -\left(1-\lambda\right)\omega_{X}\int\limits_{s\in S_{X}}\psi\left(s\right)\,d\sigma_{X}^{0}\left(s\right) ; \\
\alpha g_{X}\left(y\right) + \beta g_{Y}\left(y\right) + \gamma &= -\lambda \omega_{X}\int\limits_{s\in S_{X}}\psi\left(s\right)\,\sigma_{X}^{Y}\left[y\right]\left(s\right) -\left(1-\lambda\right)\omega_{X}\int\limits_{s\in S_{X}}\psi\left(s\right)\,d\sigma_{X}^{0}\left(s\right) 
\end{align}
\end{subequations}
\end{linenomath}
for each $x\in S_{X}$ and $y\in S_{Y}$, then the strategy $\left(\sigma_{X}^{0},\sigma_{X}^{X}\left[x\right] ,\sigma_{X}^{Y}\left[y\right]\right)$ allows $X$ to enforce the equation $\alpha\pi_{X}+\beta\pi_{Y}+\gamma =0$ for every strategy of player $Y$.
\end{theorem}

For a proof of this result, we refer the reader to Supporting Information.

However, through examples we demonstrate in \S\ref{sec:classicalDG} and \S\ref{sec:CDG} that autocratic strategies do require that player $X$ moves sufficiently often, i.e. condition Eq.~(\ref{eq:mainEquationRandom}) implicitly puts a lower bound on $\omega_{X}$.

\begin{unlabeledRemark}
Theorems \ref{thm:strictX}, \ref{thm:strictY}, and \ref{thm:random} give conditions under which $X$ can enforce $\alpha\pi_{X}+\beta\pi_{Y}+\gamma =0$ for every strategy of player $Y$. Although $X$ is using a memory-one strategy to enforce this linear relationship, we make no assumptions on $Y$'s strategy; it can be any behavioral strategy with arbitrary (even infinite) memory. For two-action (and undiscounted) games with simultaneous moves, \citet{press:PNAS:2012} show that if $X$ uses a memory-one strategy, then the strategy of $Y$ may also be assumed to be memory-one. While this result is required for the use of their ``determinant trick" to establish the existence of zero-determinant strategies, it is not needed here due to a technique of \citet{akin:Games:2015}. Further details are in Supporting Information.
\end{unlabeledRemark}

\section{Example: classical Donation Game}\label{sec:classicalDG}

While our main results hold for alternating games with generic action spaces, we first illustrate their implications for the classical, two-action Donation Game. The classical Donation Game, whose payoff matrix is given by Eq.~(\ref{eq:donationPayoffMatrix}), is based on the discrete actions of cooperate, $C$, and defect, $D$. Without discounting, initial moves do not matter and hence a memory-one strategy for player $X$ is defined by a four-tuple, $\mathbf{p}:=\left(p_{CC},p_{CD},p_{DC},p_{DD}\right)$, where $p_{xy}$ is the probability that $X$ cooperates after $X$ plays $x$ and $Y$ plays $y$ for $x,y\in\left\{C,D\right\}$. In the simultaneous-move Donation Game, \citet{press:PNAS:2012} show that for $\chi\geqslant 1$,
\begin{linenomath}
\begin{align}
\label{eq:pdiscrete}
\mathbf{p} &= \left(1-\phi\left(\chi -1\right)\left(b-c\right),\ 1-\phi\left(\chi b+c\right),\ \phi\left(b+\chi c\right),\ 0\right)
\end{align}
\end{linenomath}
unilaterally enforces the extortionate relationship $\pi_{X}=\chi\pi_{Y}$ provided that a normalization factor, $\phi$, exists.

In undiscounted (infinite) and strictly-alternating games, we know from Eqs.~(\ref{eq:mainEquationStrictI}) and (\ref{eq:mainEquationStrictII}) that player $X$ does not need to take into account who moves first when devising an autocratic strategy and, moreover, the conditions become identical to those for simultaneous games \citep{mcavoy:PNAS:2016,press:PNAS:2012}. Therefore, player $X$ can use a single strategy to enforce $\pi_{X}=\chi\pi_{Y}$ in both simultaneous and strictly-alternating games. For discounted (finite) games, however, autocratic strategies depend on whether the moves are simultaneous or strictly-alternating, but the condition on the discounting factor guaranteeing their existence, $\lambda\geqslant\left(b+\chi c\right) / \left(\chi b+c\right)$, does not.

In the undiscounted but randomly-alternating Donation Game, player $X$ moves with probability $\omega_{X}$ and player $Y$ with probability $1-\omega_{X}$ in each round. A memory-one strategy for player $X$ is given by $\mathbf{p}^{X}=\left(p_{C}^{X},p_{D}^{X}\right)$ and $\mathbf{p}^{Y}=\left(p_{C}^{Y},p_{D}^{Y}\right)$, where $p_{x}^{X}$ (resp. $p_{y}^{Y}$) denotes the probability that $X$ plays $C$ if $X$ moved $x$ (resp. $Y$ moved $y$) in the preceding round. In this game, player $X$ can enforce $\pi_{X}=\chi\pi_{Y}$ with
\begin{linenomath}
\begin{align}
\label{eq:pdiscreteAlternating}
\mathbf{p}^{X}=\left(\frac{1}{\omega_{X}}\left( 1-\phi\left(\chi b+c\right) \right),\ 0\right),\qquad \mathbf{p}^{Y}=\left(\frac{1}{\omega_{X}}\phi\left(b+\chi c\right),\ 0\right)
\end{align}
\end{linenomath}
provided that the normalization factor, $\phi$, falls within the range
\begin{linenomath}
\begin{align}
\frac{1-\omega_{X}}{\chi b+c} \leqslant \phi \leqslant \min\left\{\frac{\omega_{X}}{b+\chi c},\frac{1}{\chi b+c}\right\} .
\end{align}
\end{linenomath}
The existence of such a $\phi$ in this range requires that $X$ moves sufficiently frequently, i.e.
\begin{linenomath}
\begin{align}
\omega_{X} &\geqslant \frac{b+\chi c}{\left(\chi +1\right)\left(b+c\right)} .
\end{align}
\end{linenomath}
Otherwise, player $X$ loses control over the outcome of the game and can no longer enforce a linear payoff relationship. The autocratic strategy defined by Eq.~(\ref{eq:pdiscreteAlternating}) is unforgiving and always responds to defection with defection but more readily cooperates than its counterpart for simultaneous or strictly-alternating Donation Games, defined by Eq.~(\ref{eq:pdiscrete}), since $p_{C}^{X}=\frac{1}{\omega_{X}}p_{CD}\geqslant p_{CD}$ and $p_{C}^{Y}=\frac{1}{\omega_{X}}p_{DC}\geqslant p_{DC}$.

\section{Example: continuous Donation Game}\label{sec:CDG}

In the continuous Donation Game, the action space available to players $X$ and $Y$ is an interval $\left[0,K\right]$, which indicates a continuous range of cooperative investment levels with an upper bound $K>0$. If $X$ plays $x\in\left[0,K\right]$, she donates $b(x)$ to her opponent at a cost $c(x)$ to herself with $b(x)>c(x)$ for $x>0$ and $b(0)=c(0)=0$ \citep{killingback:AN:2002}. This game is symmetric with $f_{X}\left(s\right) =g_{Y}\left(s\right) =-c\left(s\right)$ and $f_{Y}\left(s\right) =g_{X}\left(s\right) =b\left(s\right)$.

\subsection{Extortionate, generous, and equalizer strategies}

For each variant of alternating moves, we consider three particularly important classes of autocratic strategies for the continuous Donation Game: equalizer, extortionate, and generous. An equalizer strategy is an autocratic strategy that allows $X$ to unilaterally set either $\pi_{X}=\gamma$ (self-equalizing) or $\pi_{Y}=\gamma$ (opponent-equalizing) \citep{hilbe:PNAS:2013}. In all scenarios, we show that no self-equalizing strategies exist that allow player $X$ to set $\pi_{X}=\gamma$ for $\gamma >0$. However, player $X$ can typically set the score of her opponent. Equalizer strategies are defined in the same way for alternating and simultaneous-move games, whereas extortionate and generous strategies require slightly different definitions. In the simultaneous version of the continuous Donation Game, player $X$ can enforce the linear relationship $\pi_{X}-\kappa =\chi\left(\pi_{Y}-\kappa\right)$ for any $\chi\geqslant 1$ and $0\leqslant\kappa\leqslant b\left(K\right) -c\left(K\right)$, provided $\lambda$ is sufficiently close to $1$ \citep{mcavoy:PNAS:2016}. Note that the ``baseline payoff," $\kappa$, indicates the payoff of an autocratic strategy against itself \citep{hilbe:NC:2014}. If $\chi >1$ and $\kappa =0$, then such an autocratic strategy is called ``extortionate" since it ensures that the expected payoff of player $X$ is at least that of player $Y$. Conversely, if $\kappa =b\left(K\right) -c\left(K\right)$, then this strategy is called ``generous" (or ``compliant") since it ensures that the expected payoff of player $X$ is at most that of player $Y$ \citep{stewart:PNAS:2013,hilbe:PNAS:2013}.

The bounds on $\kappa$ arise from the payoffs for mutual cooperation and mutual defection in repeated games. Of course, in the simultaneous-move game, those bounds are the same as the payoffs for mutual cooperation and defection in one-shot interactions. Discounted (finite), alternating games, on the other hand, result in asymmetric payoffs even if the underlying one-shot interaction is symmetric. For example, if player $X$ moves first in the strictly-alternating, continuous Donation Game, and if both players are unconditional cooperators, then $\pi_{X}=\left(\lambda b\left(K\right) -c\left(K\right)\right) /\left(1+\lambda\right)$ but $\pi_{Y}=\left(b\left(K\right) -\lambda c\left(K\right)\right)/\left(1+\lambda\right)$, which are not equal for discounting factors $\lambda <1$.
Thus, rather than comparing both $\pi_{X}$ and $\pi_{Y}$ to the same payoff, $\kappa$, it makes more sense to compare $\pi_{X}$ to $\kappa_{X}$ and $\pi_{Y}$ to $\kappa_{Y}$ for some $\kappa_{X}$ and $\kappa_{Y}$. Therefore, we focus on conditions that allow player $X$ to enforce $\pi_{X}-\kappa_{X}=\chi\left(\pi_{Y}-\kappa_{Y}\right)$ for $\kappa_{X}$ and $\kappa_{Y}$ within a suitable range. Note that if player $X$ enforces this payoff relation and, conversely, player $Y$ enforces $\pi_{Y}-\kappa_{Y}=\chi\left(\pi_{X}-\kappa_{X}\right)$ for some $\chi >1$, then player $X$ gets $\kappa_{X}$ and $Y$ gets $\kappa_{Y}$, which preserves the original interpretation of $\kappa$ as the ``baseline payoff" \citep{hilbe:NC:2014}. Also note that the two strategies enforcing the respective payoff relation need not be the same due to the asymmetry in payoffs, which arises from the asymmetry induced by alternating moves.

For $s,s'\in\left\{0,K\right\}$, let $\kappa_{X}^{ss'}$ and $\kappa_{Y}^{ss'}$ be the baseline payoffs to players $X$ and $Y$, respectively, when $X$ uses $s$ unconditionally and $Y$ uses $s'$ unconditionally in the repeated game. For sufficiently weak discounting factors, $\lambda$, player $X$ can enforce $\pi_{X}-\kappa_{X}=\chi\left(\pi_{Y}-\kappa_{Y}\right)$ for any alternating game if and only if
\begin{linenomath}
\begin{align}\label{eq:kappaInequalities}
\kappa_{X} = \kappa_{X} + \left( \chi\kappa_{Y}^{00} - \kappa_{X}^{00} \right) \leqslant \chi\kappa_{Y} \leqslant \kappa_{X} + \left( \chi\kappa_{Y}^{KK} - \kappa_{X}^{KK} \right) ,
\end{align}
\end{linenomath}
where $\kappa_{X}^{00}=\kappa_{Y}^{00}=0$. Eq. (\ref{eq:kappaInequalities}) implies that if player $X$ attempts to minimize player $Y$'s baseline payoff, $\kappa_{Y}$, for a fixed $\kappa_{X}$, then $\chi\kappa_{Y}=\kappa_{X}$. Hence player $X$ enforces $\pi_{X}-\kappa_{X} =\chi\left(\pi_{Y}-\kappa_{Y}\right) =\chi\pi_{Y}-\kappa_{X}$, or $\pi_{X}=\chi\pi_{Y}$. Such an autocratic strategy is called ``extortionate" since it minimizes the baseline payoff of the opponent, or, equivalently, it minimizes the difference $\chi\kappa_{Y}-\kappa_{X}$. Conversely, if player $X$ tries to maximize $Y$'s baseline payoff, then $\chi\kappa_{Y}=\kappa_{X}+\left(\chi\kappa_{Y}^{KK} - \kappa_{X}^{KK}\right)$ and $X$ enforces the equation $\pi_{X}-\kappa_{X}^{KK}=\chi\left(\pi_{Y}-\kappa_{Y}^{KK}\right)$. This type of autocratic strategy is called ``generous" since it maximizes the baseline payoff of the opponent, or, equivalently, it maximizes the difference $\chi\kappa_{Y}-\kappa_{X}$. Therefore, qualitatively speaking, in spite of the more detailed considerations necessary for alternating games, the introduction of distinct baseline payoffs for players $X$ and $Y$ does not affect the spirit in which extortionate and generous strategies are defined.

Interestingly, and somewhat surprisingly, it is possible for player $X$ to devise an autocratic strategy based on merely two distinct actions, $s_{1}$ and $s_{2}$, despite the fact that her opponent may draw on a continuum of actions \citep{mcavoy:PNAS:2016}. In strictly-alternating games, such a ``two-point" strategy adjusts $p\left(x,y\right)$ (resp. $1-p\left(x,y\right)$), the probability of playing $s_{1}$ (resp. $s_{2}$), in response to the previous moves, $x$ and $y$, while the actions $s_{1}$ and $s_{2}$ themselves remain unchanged. These strategies are particularly illustrative because they admit analytical solutions and simpler intuitive interpretations. In the following we focus first on two-point autocratic strategies for player $X$ based on the two actions of full cooperation, $K$, and defection, $0$. For each variant of alternating game, we derive stochastic two-point strategies enforcing extortionate, generous, and equalizer payoff relationships. For the more interesting case of randomly-alternating moves, we also give deterministic analogues of these strategies that use infinitely many points in the action space.

\subsection{Strictly-alternating moves; player $X$ moves first}

The baseline payoffs for full, mutual cooperation if player $X$ moves first are
\begin{linenomath}
\begin{align}
\kappa_{X}^{KK} = \frac{\lambda b\left(K\right) -c\left(K\right)}{1+\lambda} ; \qquad
\kappa_{Y}^{KK} = \frac{b\left(K\right) -\lambda c\left(K\right)}{1+\lambda} ,
\end{align}
\end{linenomath}
while the baseline payoffs for mutual defection are always $\kappa_{X}^{00}=\kappa_{Y}^{00}=0$. The function $\psi\left(s\right) :=-\chi b\left(s\right) -c\left(s\right)$ conveniently eliminates $x$ from Eq.~(\ref{eq:mainEquationStrictI}). For sufficiently long interactions or weak discounting factors, i.e.
\begin{linenomath}
\begin{align}\label{eq:strictLambdaInequality}
\lambda &\geqslant \frac{b\left(K\right) + \chi c\left(K\right)}{\chi b\left(K\right) + c\left(K\right)} ,
\end{align}
\end{linenomath}
the two-point strategy defined by
\begin{linenomath}
\begin{align}\label{eq:strictPofY}
p\left(y\right) &= \frac{\lambda\left(b\left(y\right) +\chi c\left(y\right)\right) + \left(1+\lambda\right)\left(\chi\kappa_{Y}-\kappa_{X}\right)}{\lambda^{2}\left(\chi b\left(K\right) +c\left(K\right)\right)} - \frac{1-\lambda^{2}}{\lambda^{2}} p_{0} ,
\end{align}
\end{linenomath}
allows player $X$ to unilaterally enforce $\pi_{X}-\kappa_{X}=\chi\left(\pi_{Y}-\kappa_{Y}\right)$ as long as $p_{0}$ falls within a suitable range (see Eq. (\ref{sieq:extGenStrictInitial})) and $\kappa_{X}\leqslant\chi\kappa_{Y}\leqslant\kappa_{X}+\left( \chi\kappa_{Y}^{KK} - \kappa_{X}^{KK} \right)$. Whether the autocratic strategy defined by $\left(p_{0},p\left(y\right)\right)$ is extortionate or generous depends on the choice of $\chi$, $\kappa_{X}$, and $\kappa_{Y}$. Note that $p\left(y\right)$ does not depend on player $X$'s own previous move, $x$, and hence represents an instance of a reactive strategy \citep{nowak:AAM:1990}.

Similarly, choosing $\psi(s) :=b\left(s\right)$ again eliminates $x$ from Eq.~(\ref{eq:mainEquationStrictI}) but now enables player $X$ to adopt an equalizer strategy and set her opponent's score to $\pi_{Y}=\gamma$ with $0\leqslant\gamma\leqslant \left(b\left(K\right) -\lambda c\left(K\right)\right) /\left(1+\lambda\right)$ by using
\begin{linenomath}
\begin{align}
p\left(y\right) &= \frac{\lambda c\left(y\right) + \left(1+\lambda\right)\gamma}{\lambda^{2}b\left(K\right)} - \frac{1-\lambda^{2}}{\lambda^{2}} p_{0} ,
\end{align}
\end{linenomath}
provided the initial probability of cooperation, $p_{0}$, falls within a feasible range (see Eq. (\ref{sieq:equalizerStrictInitial})). However, just like in the simultaneous-move game, player $X$ can never set her own score \citep[see][]{mcavoy:PNAS:2016}.

\subsection{Strictly-alternating moves; player $Y$ moves first}

The baseline payoffs for full cooperation if player $Y$ moves first are
\begin{linenomath}
\begin{align}
\kappa_{X}^{KK} = \frac{b\left(K\right) -\lambda c\left(K\right)}{1+\lambda} ; \qquad
\kappa_{Y}^{KK} = \frac{\lambda b\left(K\right) -c\left(K\right)}{1+\lambda} ,
\end{align}
\end{linenomath}
and, again, the baseline payoffs for mutual defection are both $0$. The scaling function $\psi\left(s\right) :=-\chi b\left(s\right) -c\left(s\right) +\chi\kappa_{Y}-\kappa_{X}$ eliminates $x$ from Eq.~(\ref{eq:mainEquationStrictII}). For sufficiently weak discounting, i.e. if Eq.~(\ref{eq:strictLambdaInequality}) holds, the autocratic, reactive strategy, that cooperates (plays $K$) with probability
\begin{linenomath}
\begin{align}
p\left(y\right) &= \frac{b\left(y\right) +\chi c\left(y\right) +\left(1+\lambda\right)\left(\chi\kappa_{Y}-\kappa_{X}\right)}{\lambda\left(\chi b\left(K\right) +c\left(K\right)\right)}
\end{align}
\end{linenomath}
after player $Y$ moved $y$, then enables player $X$ to unilaterally enforce $\pi_{X}-\kappa_{X}=\chi\left(\pi_{Y}-\kappa_{Y}\right)$ whenever $\kappa_{X}\leqslant\chi\kappa_{Y}\leqslant\kappa_{X}+\left(\chi\kappa_{Y}^{KK}-\kappa_{X}^{KK}\right)$. Again, whether $p\left(y\right)$ translates into an extortionate or generous strategy depends on $\chi$, $\kappa_{X}$, and $\kappa_{Y}$. Note that the first move of $X$ depends on simply her opponent's initial move and hence does not need to be specified separately.

Similarly, setting $\psi\left(s\right) :=b\left(s\right) -\gamma$ also eliminates $x$ from Eq.~(\ref{eq:mainEquationStrictII}) but enables player $X$ to enforce $\pi_{Y}=\gamma$ with $0\leqslant\gamma\leqslant \left(\lambda b\left(K\right) -c\left(K\right)\right) /\left(1+\lambda\right)$, which implicitly requires $\lambda\geqslant c\left(K\right) / b\left(K\right)$. This equalizer strategy plays $K$ with probability
\begin{linenomath}
\begin{align}
p\left(y\right) &= \frac{c\left(y\right) +\left(1+\lambda\right)\gamma}{\lambda b\left(K\right)}
\end{align}
\end{linenomath}
after player $Y$ played $y$. Although player $X$ can set the score of player $Y$, she cannot set her own score.

\subsection{Randomly-alternating moves}\label{subsubsec:randomCDG}

In randomly-alternating games, the average payoffs, $\pi_{X}$ and $\pi_{Y}$, for players $X$ and $Y$, respectively, depend on the probability $\omega_{X}$ with which player $X$ moves in any given round; see Eq.~(\ref{eq:payRandom}). The region spanned by feasible payoff pairs, $\left(\pi_{Y}, \pi_{X}\right)$, not only depends on $\omega_{X}$ but also on the class of strategies considered; see Fig.~\ref{fig:feasibleRegion}. 
\begin{figure}
\begin{center}
\includegraphics[scale=0.55]{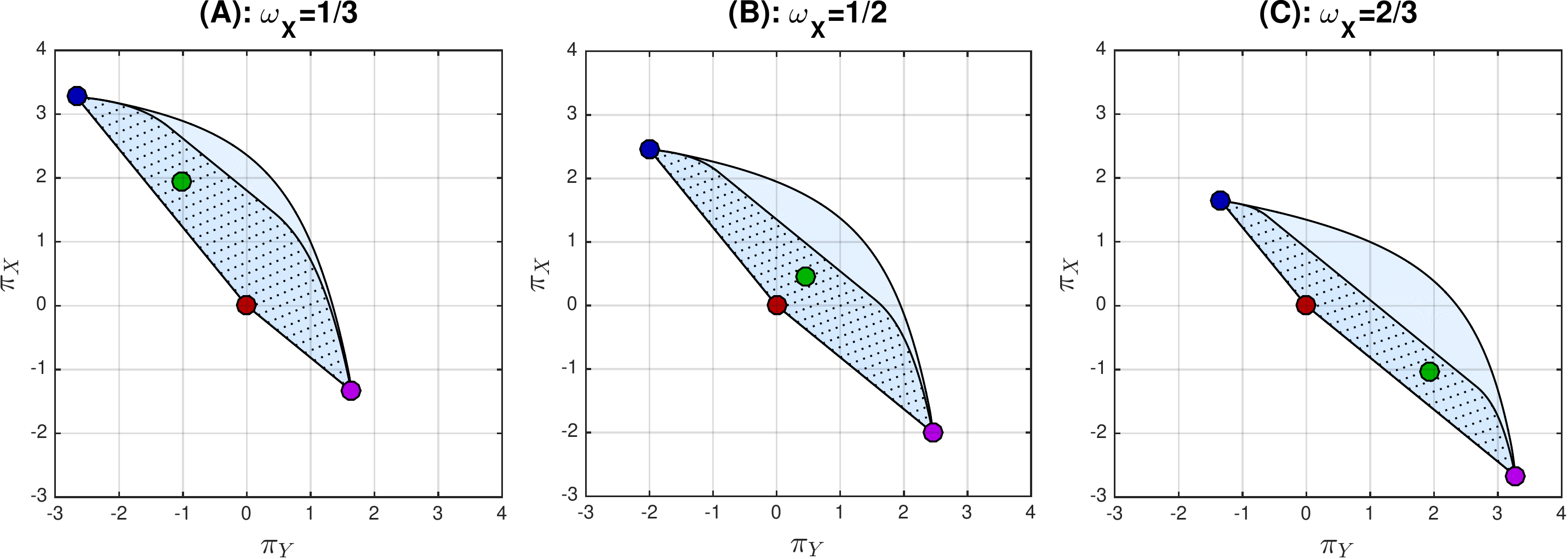}
\end{center}
\caption{Feasible payoff pairs, $\left(\pi_{Y},\pi_{X}\right)$, when $X$ uses a two-point strategy (hatched) and when $X$ uses the entire action space (light blue). The benefit function is $b\left(s\right) =5\left(1-e^{-2s}\right)$, the cost function is $c\left(s\right) =2s$, and the action spaces are $S_{X}=S_{Y}=\left[0,2\right]$ \citep[see][]{killingback:AN:2002}. The probability that player $X$ moves in any given round is (A) $\omega_{X}=1/3$, (B) $\omega_{X}=1/2$, and (C) $\omega_{X}=2/3$. In each figure, the payoffs for mutual defection ($0$) are indicated by a red point and for mutual full-cooperation ($K=2$) by a green point. The blue point marks the payoffs when $X$ defects and $Y$ fully cooperates, and the magenta point vice versa. From (A) and (C), we see that if $\omega_{X}\neq 1/2$, then the payoffs for the alternating game are typically asymmetric even though the one-shot game is symmetric.\label{fig:feasibleRegion}}
\end{figure}
In particular, the two-point strategies based on the extreme actions, $0$ and $K$, cover only a portion of the payoff region spanned by strategies utilizing the full action space, $\left[0,K\right]$. We use both two-point and deterministic autocratic strategies for illustrations as they are among the more straightforward ways in which to enforce linear payoff relationships.

\subsubsection{Two-point autocratic strategies}

Here, we focus on these two-point autocratic strategies for player $X$, which are concentrated on the two points $0$ and $K$ and are defined by (i) the probability that $X$ uses $K$ in the first round, $p_{0}$; (ii) the probability that $X$ uses $K$ following her own move in which she used $x$, $p^{X}\left(x\right)$; and (iii) the probability that $X$ uses $K$ following her opponent's move in which he used $y$, $p^{Y}\left(y\right)$.

For randomly-alternating moves, the baseline payoffs for full cooperation are
\begin{linenomath}
\begin{align}
\kappa_{X}^{KK} = \left(1-\omega_{X}\right) b\left(K\right) - \omega_{X}c\left(K\right) ; \qquad
\kappa_{Y}^{KK} = \omega_{X}b\left(K\right) - \left(1-\omega_{X}\right) c\left(K\right) ,
\end{align}
\end{linenomath}
while those for mutual defection remain both $0$. Suppose that discounting is sufficiently weak, or interactions cover sufficiently many rounds, i.e.
\begin{linenomath}
\begin{align}\label{eq:donationExtortionCondition}
\lambda &\geqslant \frac{1}{\omega_{X}}\cdot\frac{b\left(K\right) +\chi c\left(K\right)}{\left(\chi +1\right)\left(b\left(K\right) + c\left(K\right)\right)} ,
\end{align}
\end{linenomath}
and that $\kappa_{X}\leqslant\chi\kappa_{Y}\leqslant\kappa_{X}+\left(\chi\kappa_{Y}^{KK}-\kappa_{X}^{KK}\right)$. Then, the two-point autocratic strategy defined by
\begin{linenomath}
\begin{subequations}
\begin{align}
p^{X}\left(x\right) &= \frac{b\left(x\right) +\chi c\left(x\right) +\chi\kappa_{Y}-\kappa_{X}}{\lambda\omega_{X}\left(\chi +1\right)\!\left(b\left(K\right) +c\left(K\right)\!\right)} - \frac{1-\lambda}{\lambda} p_{0} ; \\
p^{Y}\left(y\right) &= p^{X}\left(y\right)
\end{align}
\end{subequations}
\end{linenomath}
enables player $X$ to unilaterally enforce $\pi_{X}-\kappa_{X}=\chi\left(\pi_{Y}-\kappa_{Y}\right)$ provided $p_{0}$ falls within a suitable range (see Eq. (\ref{sieq:extGenRandomInitial})). The scaling function $\psi\left(s\right) :=-\left(\chi +1\right)\left( b\left(s\right) + c\left(s\right) \right)$ was chosen such that $X$'s response depends on the previous action but not on which player used it.

If player $X$ is at least as likely to move in each round as is player $Y$, i.e. $\omega_{X}\geqslant 1/2$, then, for every $\chi\geqslant 1$, a sufficiently weak discounting factor exists that satisfies Eq.~(\ref{eq:donationExtortionCondition}) and $\lambda\leqslant 1$ and hence enables player $X$ to enforce $\pi_{X}-\kappa_{X}=\chi\left(\pi_{Y}-\kappa_{Y}\right)$. In particular, both extortionate and generous autocratic strategies exist for the randomly-alternating, continuous Donation Game; see Fig.~\ref{fig:twoPointSimulation}.
\begin{figure}
\begin{center}
\includegraphics[scale=0.55]{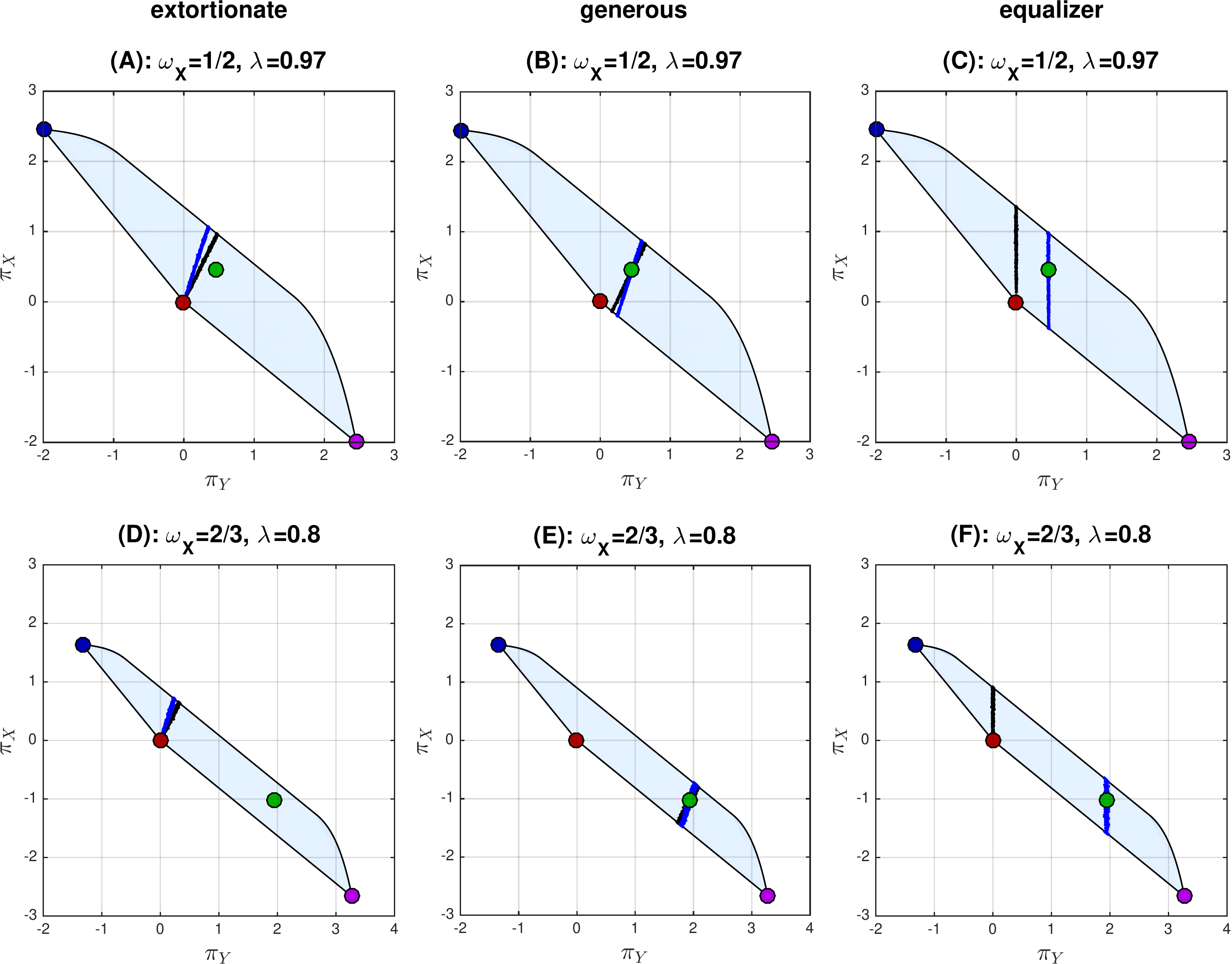}
\end{center}
\caption{Two-point extortionate, generous, and equalizer strategies for the randomly-alternating, continuous Donation Game. In each panel, the red (resp. green) point indicates the payoffs for mutual defection (resp. full cooperation). The blue point gives the payoffs when $X$ defects and $Y$ fully cooperates in every round, and the magenta point vice versa. In the top row, both players move with equal probability in a given round ($\omega_{X}=1/2$), whereas in the bottom row player $X$ moves twice as often as player $Y$ ($\omega_{X}=2/3$). The extortionate strategies in (A) and (D) enforce $\pi_{X}=\chi\pi_{Y}$, while the generous strategies in (B) and (E) enforce $\pi_{X}-\kappa_{X}^{KK}=\chi\left(\pi_{Y}-\kappa_{Y}^{KK}\right)$ with $\chi =2$ (black) and $\chi =3$ (blue). The equalizer strategies in (C) and (F) enforce $\pi_{Y}=\gamma$ with $\gamma =\kappa_{Y}^{00}=0$ (black) and $\gamma =\kappa_{Y}^{KK}$ (blue). The simulation data in each panel show the average payoffs, $\left(\pi_{Y},\pi_{X}\right)$, for player $X$'s two-point strategy against $1000$ random memory-one strategies for player $Y$. The benefit function is $b\left(s\right) =5\left(1-e^{-2s}\right)$ and the cost function is $c\left(s\right) =2s$ for action spaces $S_{X}=S_{Y}=\left[0,2\right]$. \label{fig:twoPointSimulation}}
\end{figure}

Similarly, an equalizing two-point strategy for player $X$ can ensure $\pi_{Y}=\gamma$ for any
\begin{linenomath}
\begin{align}
\label{eq:gammaRangeRandom}
0 &\leqslant \gamma \leqslant \omega_{X}b\left(K\right) -\left(1-\omega_{X}\right) c\left(K\right) .
\end{align}
\end{linenomath}
$X$ can enforce such a relationship using the two-point strategy defined by
\begin{linenomath}
\begin{subequations}
\begin{align}
p^{X}\left(x\right) &= \frac{c\left(x\right) + \gamma}{\lambda\,\omega_{X}\left( b\left(K\right) + c\left(K\right) \right)} - \frac{1-\lambda}{\lambda} p_{0} ; \\
p^{Y}\left(y\right) &= p^{X}\left(y\right)
\end{align}
\end{subequations}
\end{linenomath}
provided $p_{0}$ falls within a specified range (see Eq. (\ref{sieq:equalizerRandomInitial})). Note that player $X$ is unable to unilaterally set the payoff of player $Y$ to anything below unconditional defection, $0$, and beyond unconditional cooperation, $\omega_{X}b\left(K\right) -\left(1-\omega_{X}\right) c\left(K\right)$; see Supporting Information for further details. Moreover, it must be true that $\lambda\omega_{X}b\left(K\right)\geqslant\left(1-\lambda\omega_{X}\right) c\left(K\right)$. For $\omega_{X}=1/2$, the discounting factor, $\lambda$, must therefore satisfy $\lambda\geqslant 2c\left(K\right)/\left(b\left(K\right) +c\left(K\right)\right)$, which enables player $X$ to set $Y$'s score to anything between $0$ and $\left(b\left(K\right) -c\left(K\right)\right) /2$. 
In the limit where player $X$ moves exclusively, $\omega_{X}\rightarrow 1$, player $Y$'s score can be set to at most $b\left(K\right)$, which, itself, is clear from the definition of the continuous Donation Game.

Although player $X$ can unilaterally set $Y$'s score, she cannot set her own score to anything above $0$, and, for $\omega_{X}\leqslant 1/2$, she cannot set her own score to anything at all. For sufficiently large $\lambda\omega_{X}$, player $X$ can guarantee herself non-positive payoffs using an autocratic strategy; see Supporting Information. However, strategies enforcing a return that is at most $0$ may be of limited use since a player can always do at least as well through unconditional defection. In contrast, in the simultaneous version of the continuous Donation Game, player $X$ can never set her own score \citep{mcavoy:PNAS:2016}. This difference is not that surprising: even though player $X$ can exert control over randomly-alternating games for large $\omega_{X}$, the structure of the continuous Donation Game precludes her from providing herself positive payoffs through actions of her own.

\subsubsection{Deterministic autocratic strategies}

Deterministic strategies, for which $X$ reacts to the previous move by playing an action with certainty (rather than probabilistically), cover a broader range of feasible payoffs (see Fig. \ref{fig:feasibleRegion}) than do two-point strategies \citep[see also][]{mcavoy:PNAS:2016}. A simple example of a deterministic strategy is tit-for-tat, which cooperates in the first round and subsequently copies the opponent's previous move \citep{axelrod:BB:1984}.

A deterministic strategy for a randomly-alternating game consists of (i) an initial action, $x_{0}$; (ii) a reaction function to one's own move, $r^{X}:S_{X}\rightarrow S_{X}$; and (iii) a reaction function to the opponent's move, $r^{Y}:S_{Y}\rightarrow S_{X}$. Here, we give examples of deterministic extortionate, generous, and equalizer strategies for the continuous Donation Game. For example, player $X$ can enforce $\pi_{X}-\kappa_{X}=\chi\left(\pi_{Y}-\kappa_{Y}\right)$ by using
\begin{linenomath}
\begin{subequations}
\begin{align}
r^{X}\left(x\right) &= \left(b+c\right)^{-1}\left(\frac{b\left(x\right) +\chi c\left(x\right) + \chi\kappa_{Y}-\kappa_{X}-\left(1-\lambda\right)\omega_{X}\left(\chi +1\right)\Big(b\left(x_{0}\right) +c\left(x_{0}\right)\Big)}{\lambda\omega_{X}\left(\chi +1\right)}\right) ; \\
r^{Y}\left(y\right) &= r^{X}\left(y\right) ,
\end{align}
\end{subequations}
\end{linenomath}
where $\left(b+c\right)^{-1}\left(\cdots\right)$ denotes the inverse of the function $b\left(s\right) +c\left(s\right)$, provided the initial action, $x_{0}$, is chosen appropriately. For instance, if $\kappa_{X}=\kappa_{X}^{00}=0$ and $\kappa_{Y}=\kappa_{Y}^{00}=0$, then $X$ may use $x_{0}=0$ to enforce the extortionate relationship $\pi_{X}=\chi\pi_{Y}$. If $\kappa_{X}=\kappa_{X}^{KK}$ and $\kappa_{Y}=\kappa_{Y}^{KK}$, then $X$ may use $x_{0}=K$ to enforce the generous relationship $\pi_{X}-\kappa_{X}^{KK}=\chi\left(\pi_{Y}-\kappa_{Y}^{KK}\right)$. In both cases, $\lambda$ must satisfy Eq. (\ref{eq:donationExtortionCondition}) for $r^{X}$ and $r^{Y}$ to be well-defined reaction functions.
Similarly, $X$ can unilaterally equalize $Y$'s payoff to $\pi_{Y}=\gamma$ by using
\begin{linenomath}
\begin{subequations}
\begin{align}
r^{X}\left(x\right) &= \left(b+c\right)^{-1}\left(\frac{c\left(x\right) + \gamma - \left(1-\lambda\right)\omega_{X}\left( b\left(x_{0}\right) +c\left(x_{0}\right)\right)}{\lambda\omega_{X}}\right) ; \\
r^{Y}\left(y\right) &= r^{X}\left(y\right) .
\end{align}
\end{subequations}
\end{linenomath}
If $\gamma =\kappa_{Y}^{00}=0$, then player $X$ may use $x_{0}=0$; if $\gamma =\kappa_{Y}^{KK}$, then $X$ may use $x_{0}=K$.
The feasible payoff regions and simulation data for each of these classes of autocratic strategies are given in Fig. \ref{fig:deterministicSimulation}.

\begin{figure}
\begin{center}
\includegraphics[scale=0.55]{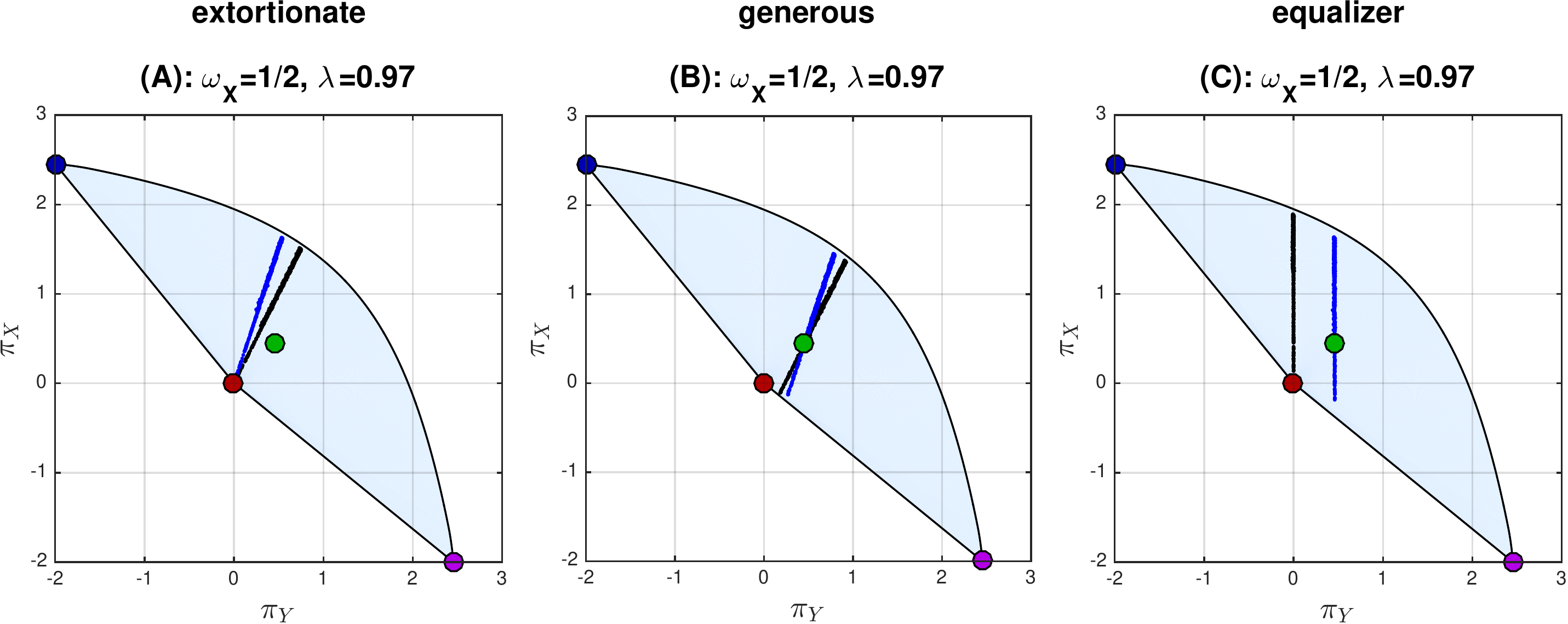}
\end{center}
\caption{Deterministic extortionate, generous, and equalizer strategies for the randomly-alternating, continuous Donation Game. In each round, both players have the same probability of moving ($\omega_{X}=1/2$). In (A), extortionate strategies enforce $\pi_{X}=\chi\pi_{Y}$ and in (B), generous strategies enforce $\pi_{X}-\kappa_{X}^{KK}=\chi\left(\pi_{Y}-\kappa_{Y}^{KK}\right)$ with $\chi =2$ (black) and $\chi =3$ (blue). In (C), equalizer strategies enforce $\pi_{Y}=\gamma$ with $\gamma =\kappa_{Y}^{00}=0$ (black) and $\gamma =\kappa_{Y}^{KK}$ (blue). Since deterministic strategies utilize a larger portion of the action space than two-point strategies, the players can attain a broader range of payoff pairs, $\left(\pi_{Y},\pi_{X}\right)$ (c.f. Fig. \ref{fig:twoPointSimulation}). The simulation data in each panel shows the average payoffs, $\left(\pi_{Y},\pi_{X}\right)$, for $X$'s deterministic strategy against $1000$ randomly-chosen, memory-one strategies of the opponent. The benefit function is $b\left(s\right) =5\left(1-e^{-2s}\right)$ and the cost function is $c\left(s\right) =2s$ for $s\in\left[0,2\right]$. \label{fig:deterministicSimulation}}
\end{figure}

What is noteworthy about these strategies is that they require only the last move and not who played it. In other words, a player using one of these strategies responds to a move by herself in exactly the same way as she responds to a move by her opponent. Although a player never knows in advance when she will move in a randomly-alternating game, she can still enforce extortionate, generous, and equalizer relationships on payoffs by playing an action that is uniquely determined by the most recent action of the game.

\section{Discussion}

Repeated games likely rank among the best-studied topics in game theory, and the resulting insights have been instrumental for our understanding of strategic behavioral patterns. For this reason, it came as all the more of a surprise when \citet{press:PNAS:2012} reported a new class of ``zero-determinant" strategies, which enable players to exert unprecedented control in repeated interactions. However, notwithstanding decades of extensive literature on repeated games, alternating interactions have received very little attention when compared to their simultaneous counterparts. This emphasis is particularly puzzling because many, if not most, social encounters among plants or animals (including humans) that unfold over several rounds seem better captured by alternating actions of the interacting agents. Moreover, even within the realm of alternating games, it is often assumed that individual turns alternate strictly rather than randomly \citep{frean:PRSB:1994,hauert:JTB:1998,neill:JTB:2001,zagorsky:PLOSONE:2013}.

Here, we introduce autocratic strategies, a generalization of zero-determinant strategies, for alternating games. Due to similarities with simultaneous-move games, it is perhaps unsurprising that autocratic strategies also exist for strictly-alternating games. However, even so, the continuous Donation Game demonstrates that the autocratic strategies themselves depend on the timing of the players' moves. What is more surprising, and even unexpected, is the fact that autocratic strategies exist for randomly-alternating games as well. This extension exemplifies the surprising robustness of autocratic strategies by relaxing the original assumptions in three important ways: (i) to allow for discounted payoffs, i.e. to consider finite numbers of rounds in each interaction \citep{hilbe:GEB:2015}; (ii) to extend the action set from two distinct actions to infinite action spaces \citep{mcavoy:PNAS:2016}; and now (iii) to admit asynchronous decisions and, in particular, randomly-alternating ones. The latter even includes asymmetric scenarios where one player moves, on average, more frequently than the other. Under this far more generic setup we demonstrate that autocratic strategies still exist and enable players to enforce extortionate, generous, and equalizer relationships with their opponent.

In the strictly-alternating, continuous Donation Game, autocratic strategies exist for player $X$ provided that the discounting factor, $\lambda$, is sufficiently weak, or, equivalently, that interactions span sufficiently many rounds; see Eq.~(\ref{eq:strictLambdaInequality}). Interestingly, the condition on $\lambda$ does not depend on whether player $X$ moves first or second and is even identical to the corresponding condition in the synchronous game \citep{mcavoy:PNAS:2016}. In the absence of discounting, $\lambda =1$, the same strategy enforces, for instance, an extortionate payoff relationship in simultaneous games as well as alternating games and regardless of whether or not player $X$ moved first. We demonstrate this phenomenon for the classical and continuous Donation Games in \S\ref{sec:classicalDG} and \S\ref{sec:CDG}, respectively.

The condition for the existence of autocratic strategies in the randomly-alternating game, Eq.~(\ref{eq:donationExtortionCondition}), is similar to that of the strictly-alternating (and simultaneous) games, although slightly stronger. Not surprisingly, this condition depends on the probability that player $X$ moves in a given round, $\omega_{X}$. For each type of alternating game, we give examples of simple two-point autocratic strategies in which player $X$'s actions are restricted to $0$ (defect) and $K$ (fully cooperate). Although $X$ can enforce any extortionate, generous, or equalizer payoff relationship in the continuous Donation Game using a two-point strategy, a larger region of feasible payoffs is attainable if $X$ uses a deterministic autocratic strategy (see Fig.~\ref{fig:feasibleRegion}).

While autocratic strategies undoubtedly mark important behavioral patterns, their importance in an evolutionary context is still debated: extortionate strategies perform poorly \citep{adami:NC:2013}, whereas generous strategies perform much better \citep{stewart:PNAS:2013}. In fact, a generous strategy against itself represents a Nash equilibrium in the simultaneous, two-action Prisoner's Dilemma \citep{hilbe:GEB:2015}. However, for extensions to continuous action spaces, such as the continuous Donation Game, even a generous strategy with full mutual cooperation is not necessarily a Nash equilibrium \citep{mcavoy:PNAS:2016}. Similar considerations for alternating games are further nuanced because they naturally introduce asymmetries in payoffs for the two players, even if the underlying interaction is symmetric and both players follow the same strategy. In fact, this asymmetry holds for any strictly-alternating game with discounting factor $\lambda<1$ because then it matters which player moved first. Similarly, in randomly-alternating games, the payoffs typically depend on the probability $\omega_X$ with which player $X$ moves and hence differ if $\omega_X\neq 1/2$. Consequently, even if player $X$ and $Y$ adopt the same autocratic strategy, then player $X$ does not necessarily enforce the same linear relationship on payoffs as player $Y$, which complicates the notion of equilibria both in the sense of Nash as well as rest points of the evolutionary dynamics.

Among alternating games, the randomly-alternating ones represent perhaps the most promising and relevant setup from a biological perspective \citep[see][]{nowak:JTB:1994}. In the continuous Donation Game, autocratic strategies exist even if the probability that player $X$ moves in a given round differs from that of player $Y$ (i.e. $\omega_{X}\neq 1/2$). Of course, $\omega_X$ must be large enough to ensure that player $X$ is capable of exerting sufficient control over the game to pursue an autocratic strategy. For $\omega_{X}>1/2$, this condition always holds in the continuous Donation Game but might also apply under weaker conditions. Interestingly, such asymmetries easily arise from dominance hierarchies. For example, in bouts of social grooming between primates \citep{foster:AJP:2009}, subordinate individuals, $X$, typically groom dominant individuals, $Y$, more frequently than vice versa and hence $\omega_{X}>1/2$. As a consequence, the subordinate player has more autocratic strategies available to impact social grooming than does the dominant player. Thus, autocratic strategies can be particularly useful for exerting control over asymmetric interactions. This observation marks not only an important distinction between autocratic strategies for synchronous and alternating games but also promises interesting applications to biologically-relevant interactions.

\section*{Acknowledgments}

The authors are grateful to Christian Hilbe for helpful conversations and useful comments. A. M. and C. H. thank the Max Planck Institute for Evolutionary Biology and A. M. the Program for Evolutionary Dynamics for their hospitality while this paper was written. The authors acknowledge financial support from the Natural Sciences and Engineering Research Council of Canada (NSERC), grant RGPIN-2015-05795.

The authors declare no competing financial interests.

\newpage

\setcounter{section}{0}
\setcounter{equation}{0}
\setcounter{figure}{0}
\renewcommand{\thesection}{SI.\arabic{section}}
\renewcommand{\theequation}{SI.\arabic{equation}}
\renewcommand{\thefigure}{SI.\arabic{figure}}

Here, we prove our main results for each type of alternating game (strictly- and randomly-alternating moves). By a measurable space, we mean a set, $\mathcal{X}$, equipped with a $\sigma$-algebra of subsets, $\mathcal{F}\left(\mathcal{X}\right)$, although we usually suppress $\mathcal{F}\left(\mathcal{X}\right)$. The notation $\Delta\left(\mathcal{X}\right)$ indicates the space of all probability measures on $\mathcal{X}$, i.e. the set of all measures, $\mu :\mathcal{F}\left(\mathcal{X}\right)\rightarrow\left[0,\infty\right)$, with $\mu\left(\mathcal{X}\right) =1$. All functions are bounded and measurable.

\section{Strictly-alternating games}\label{si:subsec:strict}

Let $S_{X}$ and $S_{Y}$ be the action spaces available to players $X$ and $Y$, respectively. We assume that these spaces are measurable, but otherwise we impose no restrictions on them. Let $f_{X}\left(x\right)$ and $f_{Y}\left(x\right)$ be the payoffs to players $X$ and $Y$, respectively, when $X$ moves $x\in S_{X}$. Similarly, let $g_{X}\left(y\right)$ and $g_{Y}\left(y\right)$ be the payoffs to players $X$ and $Y$, respectively, when $Y$ moves $y\in S_{Y}$. If $\lambda$ is the discounting factor, then one may compress a pair of rounds in which $X$ moves first and $Y$ moves second in order to form two-round payoff functions,
\begin{linenomath}
\begin{subequations}
\begin{align}
u_{X}\left(x,y\right) &:= f_{X}\left(x\right) + \lambda g_{X}\left(y\right) ; \\
u_{Y}\left(x,y\right) &:= f_{Y}\left(x\right) + \lambda g_{Y}\left(y\right) .
\end{align}
\end{subequations}
\end{linenomath}
In each of these two-round payoff functions, the payoff from player $Y$'s move is discounted by a factor of $\lambda$ to account for the time difference or, equivalently, for the probability that the interaction ends before player $Y$'s turn.

Due to the differences in the expressions for the average payoffs when $X$ moves first and when $Y$ moves first, respectively, we treat each of these cases separately in our study of autocratic strategies.

\subsection{$X$ moves first}\label{si:subsubsec:XmovesFirst}
If player $X$ moves first, then the entire sequence of play can be grouped into two-round pairs in which $X$ moves first and $Y$ moves second. More specifically, if $\left(x_{0},y_{1},x_{2},y_{3},\dots\right)$ is the sequence of play, then this sequence may be rewritten as $\left( \left(x_{0},y_{1}\right) , \left(x_{2},y_{3}\right) ,\dots \right)$. When written in this manner, one may use $u_{X}$ to express the average payoff to player $X$ for this sequence as
\begin{linenomath}
\begin{align}
\pi_{X} &= \left(1-\lambda\right)\left[\sum_{t=0}^{\infty}\lambda^{2t}f_{X}\left(x_{2t}\right) + \sum_{t=0}^{\infty}\lambda^{2t+1}g_{X}\left(y_{2t+1}\right)\right] \nonumber \\
&= \left(1-\lambda\right)\sum_{t=0}^{\infty}\lambda^{2t}\Big( f_{X}\left(x_{2t}\right) + \lambda g_{X}\left(y_{2t+1}\right) \Big) \nonumber \\
&= \left(1-\lambda\right)\sum_{t=0}^{\infty}\lambda^{2t}u_{X}\left(x_{2t},y_{2t+1}\right) .
\end{align}
\end{linenomath}
Similarly, the average payoff to player $Y$ is $\pi_{Y}=\left(1-\lambda\right)\sum_{t=0}^{\infty}\lambda^{2t}u_{Y}\left(x_{2t},y_{2t+1}\right)$.

A time-$T$ history indicates the sequence of play from time $t=0$ until (but not including) time $t=T$ and is an element of $\mathcal{H}^{T}:=\prod_{t=0}^{T-1}\mathcal{H}_{t}^{T}$, where
\begin{linenomath}
\begin{align}\label{sieq:historyStrictI}
\mathcal{H}_{t}^{T} &:= \begin{cases}S_{X} & t\textrm{ is even}, \\ S_{Y} & t\textrm{ is odd},\end{cases}
\end{align}
\end{linenomath}
for $0\leqslant t\leqslant T-1$. For $T=0$, we let $\mathcal{H}^{0}:=\left\{\varnothing\right\}$, where $\varnothing$ is the ``empty history," which indicates that the game has not yet begun. A behavioral strategy defines a player's actions (probabilistically) for any history of play leading up to the current move \citep[see][]{fudenberg:MIT:1991}. That is, behavioral strategies for players $X$ and $Y$, respectively, may be written in terms of the space of histories as maps,
\begin{linenomath}
\begin{subequations}
\begin{align}
\sigma_{X} &: \bigsqcup_{T\geqslant 0}\mathcal{H}^{2T} \longrightarrow \Delta\left(S_{X}\right) ; \label{eq:behavioralX} \\
\sigma_{Y} &: \bigsqcup_{T\geqslant 0}\mathcal{H}^{2T+1} \longrightarrow \Delta\left(S_{Y}\right) , \label{eq:behavioralY}
\end{align}
\end{subequations}
\end{linenomath}
where $\sqcup$ denotes the disjoint union operator, and $\Delta\left(S_{X}\right)$ and $\Delta\left(S_{Y}\right)$ denote the space of probability measures on $S_{X}$ and $S_{Y}$, respectively. These strategies may be written together more compactly as a map
\begin{linenomath}
\begin{align}
\sigma &: \mathcal{H} := \bigsqcup_{T\geqslant 0}\mathcal{H}^{T} \longrightarrow \Delta\left(S_{X}\right)\sqcup\Delta\left(S_{Y}\right) \nonumber \\
&: h^{T} \longmapsto \begin{cases}\sigma_{X}\left[h^{T}\right] & T\textrm{ is even}, \\ \sigma_{Y}\left[h^{T}\right] & T\textrm{ is odd}.\end{cases}
\end{align}
\end{linenomath}
Using $\sigma$, we define a sequence of measures, $\left\{\mu_{t}\right\}_{t\geqslant 0}$, on $\mathcal{H}^{t+1}$ as follows: For $h^{T}=\left(h_{0}^{T},h_{1}^{T},\dots ,h_{T-1}^{T}\right)\in\mathcal{H}^{T}$ and $0\leqslant t\leqslant T-1$, let $h_{\leqslant t}^{T}=\left(h_{0}^{T},h_{1}^{T},\dots ,h_{t}^{T}\right)\in\mathcal{H}^{t+1}$. For $E'\in\mathcal{F}\left(\mathcal{H}^{t}\right)$ and $E\in\mathcal{F}\left(\mathcal{H}_{t}^{t+1}\right)$, let
\begin{linenomath}
\begin{align}\label{sieq:muDef}
\mu_{t}\left(E'\times E\right) &:= \int\limits_{h^{t}\in E'} \sigma\left(h^{t},E\right) \,d\sigma\left(h_{\leqslant t-2}^{t},h_{t-1}^{t}\right)\cdots\,d\sigma\left(h_{\leqslant 0}^{t},h_{1}^{t}\right)\,d\sigma\left(\varnothing ,h_{0}^{t}\right) .
\end{align}
\end{linenomath}
For $0\leqslant k\leqslant t$, let $\nu_{t}^{k}$ be the measure on $\prod_{i=t-k}^{t}\mathcal{H}_{i}^{i+1}$, which, for $E\in\mathcal{F}\left(\prod_{i=t-k}^{t}\mathcal{H}_{i}^{i+1}\right)$, is defined as
\begin{linenomath}
\begin{align}\label{sieq:nuDef}
\nu_{t}^{k}\left(E\right) &:= \mu_{t}\left(\mathcal{H}^{t-k}\times E\right) .
\end{align}
\end{linenomath}

In a $\left(2T+2\right)$-round game (rounds $0$ through $2T+1$), the expected payoff to player $X$ is
\begin{linenomath}
\begin{align}
\pi_{X}^{2T+2} &:= \int\limits_{h^{2T+2}\in\mathcal{H}^{2T+2}}\left[\left(\frac{1-\lambda}{1-\lambda^{2T+1}}\right)\sum_{t=0}^{T}\lambda^{2t}u_{X}\left(h_{2t}^{2T+2},h_{2t+1}^{2T+2}\right)\right] \nonumber \\
&\quad\quad\quad\quad \,d\sigma\left(h_{\leqslant 2T}^{2T+2},h_{2T+1}^{2T+2}\right)\cdots\,d\sigma\left(h_{\leqslant 0}^{2T+2},h_{1}^{2T+2}\right)\,d\sigma\left(\varnothing, h_{0}^{2T+2}\right) \nonumber \\
&= \left(\frac{1-\lambda}{1-\lambda^{2T+1}}\right)\sum_{t=0}^{T}\lambda^{2t}\int\limits_{h^{2T+2}\in\mathcal{H}^{2T+2}} u_{X}\left(h_{2t}^{2T+2},h_{2t+1}^{2T+2}\right) \nonumber \\
&\quad\quad\quad\quad \,d\sigma\left(h_{\leqslant 2T}^{2T+2},h_{2T+1}^{2T+2}\right)\cdots\,d\sigma\left(h_{\leqslant 0}^{2T+2},h_{1}^{2T+2}\right)\,d\sigma\left(\varnothing, h_{0}^{2T+2}\right) \nonumber \\
&= \left(\frac{1-\lambda}{1-\lambda^{2T+1}}\right)\sum_{t=0}^{T}\lambda^{2t}\int\limits_{h^{2t+2}\in\mathcal{H}^{2t+2}} u_{X}\left(h_{2t}^{2t+2},h_{2t+1}^{2t+2}\right) \nonumber \\
&\quad\quad\quad\quad \,d\sigma\left(h_{\leqslant 2t}^{2t+2},h_{2t+1}^{2t+2}\right)\cdots\,d\sigma\left(h_{\leqslant 0}^{2t+2},h_{1}^{2t+2}\right)\,d\sigma\left(\varnothing, h_{0}^{2t+2}\right) \nonumber \\
&= \left(\frac{1-\lambda}{1-\lambda^{2T+1}}\right)\sum_{t=0}^{T}\lambda^{2t}\int\limits_{\left(h_{2t}^{2t+2},h_{2t+1}^{2t+2}\right)\in\mathcal{H}_{2t}^{2t+2}\times\mathcal{H}_{2t+1}^{2t+2}} u_{X}\left(h_{2t}^{2t+2},h_{2t+1}^{2t+2}\right)\,d\nu_{2t+1}^{1}\left(h_{2t}^{2t+2},h_{2t+1}^{2t+2}\right) \nonumber \\
&= \left(\frac{1-\lambda}{1-\lambda^{2T+1}}\right)\sum_{t=0}^{T}\lambda^{2t}\int\limits_{\left(x,y\right)\in S_{X}\times S_{Y}} u_{X}\left(x,y\right)\,d\nu_{2t+1}^{1}\left(x,y\right) .
\end{align}
\end{linenomath}
In particular, the limit
\begin{linenomath}
\begin{align}\label{sieq:objectiveXStrictI}
\pi_{X} &:= \lim_{T\rightarrow\infty}\pi_{X}^{2T+2} = \left(1-\lambda\right)\sum_{t=0}^{\infty}\lambda^{2t}\int\limits_{\left(x,y\right)\in S_{X}\times S_{Y}} u_{X}\left(x,y\right)\,d\nu_{2t+1}^{1}\left(x,y\right)
\end{align}
\end{linenomath}
exists since $f_{X}$ and $g_{X}$ (and hence $u_{X}$) are bounded. Similarly, we define
\begin{linenomath}
\begin{align}\label{sieq:objectiveYStrictI}
\pi_{Y} &:= \left(1-\lambda\right)\sum_{t=0}^{\infty}\lambda^{2t}\int\limits_{\left(x,y\right)\in S_{X}\times S_{Y}} u_{Y}\left(x,y\right)\,d\nu_{2t+1}^{1}\left(x,y\right) .
\end{align}
\end{linenomath}

Our main technical lemma is an analogue of Lemma 3.1 of \citet{akin:Games:2015}:
\begin{lemma}\label{lem:mainLemmaStrictI}
For any memory-one strategy, $\sigma_{X}\left[x,y\right]$, and any $E\in\mathcal{F}\left(S_{X}\right)$,
\begin{linenomath}
\begin{align}
\sum_{t=0}^{\infty}\lambda^{2t}\int\limits_{\left(x,y\right)\in S_{X}\times S_{Y}}\Big[ \chi_{E\times S_{Y}}\left(x,y\right) - \lambda^{2}\sigma_{X}\left[x,y\right]\left(E\right)\Big]\,d\nu_{2t+1}^{1}\left(x,y\right) &= \sigma_{X}^{0}\left(E\right) ,
\end{align}
\end{linenomath}
where $\sigma_{X}^{0}:=\sigma_{X}\left[\varnothing\right]$ is the initial action of player $X$.
\end{lemma}
\begin{proof}
By the definition of the measures $\left\{\nu_{t}^{k}\right\}_{t\geqslant 0}^{0\leqslant k\leqslant t}$, we have
\begin{linenomath}
\begin{subequations}
\begin{align}
\int\limits_{\left(x,y\right)\in S_{X}\times S_{Y}}\chi_{E\times S_{Y}}\left(x,y\right)\,d\nu_{2t+1}^{1}\left(x,y\right) &= \nu_{2t}^{0}\left(E\right) ; \\
\int\limits_{\left(x,y\right)\in S_{X}\times S_{Y}}\sigma_{X}\left[x,y\right]\left(E\right)\,d\nu_{2t+1}^{1}\left(x,y\right) &= \nu_{2t+2}^{0}\left(E\right) .
\end{align}
\end{subequations}
\end{linenomath}
Therefore, it follows that
\begin{linenomath}
\begin{align}
\sum_{t=0}^{\infty}&\lambda^{2t}\int\limits_{\left(x,y\right)\in S_{X}\times S_{Y}}\Big[ \chi_{E\times S_{Y}}\left(x,y\right) - \lambda^{2}\sigma_{X}\left[x,y\right]\left(E\right)\Big]\,d\nu_{2t+1}^{1}\left(x,y\right) \nonumber \\
&= \sum_{t=0}^{\infty}\lambda^{2t}\Big(\nu_{2t}^{0}\left(E\right) -\lambda^{2}\nu_{2t+2}^{0}\left(E\right)\Big) \nonumber \\
&= \nu_{0}^{0}\left(E\right) - \lim_{t\rightarrow\infty}\lambda^{2t+2}\nu_{2t+2}^{0}\left(E\right) \nonumber \\
&= \nu_{0}^{0}\left(E\right) \nonumber \\
&= \sigma_{X}^{0}\left(E\right) ,
\end{align}
\end{linenomath}
which completes the proof.
\end{proof}

\begin{proposition}\label{prop:mainPropStrictI}
For any bounded, measurable function, $\psi :S_{X}\rightarrow\mathbb{R}$,
\begin{linenomath}
\begin{align}
\sum_{t=0}^{\infty}\lambda^{2t}\int\limits_{\left(x,y\right)\in S_{X}\times S_{Y}}\left[ \psi\left(x\right) - \lambda^{2}\int\limits_{s\in S_{X}}\psi\left(s\right)\,d\sigma_{X}\left[x,y\right]\left(s\right)\right]\,d\nu_{2t+1}^{1}\left(x,y\right) &= \int\limits_{s\in S_{X}}\psi\left(s\right)\,d\sigma_{X}^{0}\left(s\right) .
\end{align}
\end{linenomath}
\end{proposition}
\begin{proof}
The result follows from Lemma \ref{lem:mainLemmaStrictI} and the dominated convergence theorem. We do not include the details here; the argument is the same as the proof of Proposition 1 of \citep{mcavoy:PNAS:2016}.
\end{proof}

Using Proposition \ref{prop:mainPropStrictI}, we now prove the first of our main results for strictly-alternating games:
\begin{strictXTheorem}[Autocratic strategies for strictly-alternating games in which $X$ moves first]
Suppose that
\begin{linenomath}
\begin{align}\label{sieq:conditionStrictI}
\Big( \alpha &f_{X}\left(x\right) + \beta f_{Y}\left(x\right) + \gamma \Big) + \lambda\Big( \alpha g_{X}\left(y\right) + \beta g_{Y}\left(y\right) + \gamma \Big) \nonumber \\
&= \psi\left(x\right) - \lambda^{2}\int\limits_{s\in S_{X}}\psi\left(s\right)\,d\sigma_{X}\left[x,y\right]\left(s\right) - \left(1-\lambda^{2}\right)\int\limits_{s\in S_{X}}\psi\left(s\right)\,d\sigma_{X}^{0}\left(s\right)
\end{align}
\end{linenomath}
holds for some bounded $\psi$ and for each $x\in S_{X}$ and $y\in S_{Y}$. Then, if player $X$ moves first, the pair $\left(\sigma_{X}^{0},\sigma_{X}\left[x,y\right]\right)$ allows $X$ to enforce the equation $\alpha\pi_{X}+\beta\pi_{Y}+\gamma =0$ for every strategy of player $Y$.
\end{strictXTheorem}
\begin{proof}
If Eq. (\ref{sieq:conditionStrictI}) holds, then, by Proposition \ref{prop:mainPropStrictI} and Eqs. (\ref{sieq:objectiveXStrictI}) and (\ref{sieq:objectiveYStrictI}),
\begin{linenomath}
\begin{align}
\alpha\pi_{X}+&\beta\pi_{Y}+\gamma +\left(1-\lambda\right)\int\limits_{s\in S_{X}}\psi\left(s\right)\,d\sigma_{X}^{0}\left(s\right) \nonumber \\
&= \left(1-\lambda\right)\sum_{t=0}^{\infty}\lambda^{2t}\int\limits_{\left(x,y\right)\in S_{X}\times S_{Y}}\left[ \psi\left(x\right) - \lambda^{2}\int\limits_{s\in S_{X}}\psi\left(s\right)\,d\sigma_{X}\left[x,y\right]\left(s\right)\right]\,d\nu_{2t+1}^{1}\left(x,y\right) \nonumber \\
&= \left(1-\lambda\right)\int\limits_{s\in S_{X}}\psi\left(s\right)\,d\sigma_{X}^{0}\left(s\right) ,
\end{align}
\end{linenomath}
and it follows that $\alpha\pi_{X}+\beta\pi_{Y}+\gamma =0$.
\end{proof}

\subsection{$Y$ moves first}\label{si:subsubsec:YmovesFirst}
If $Y$ moves first, then a sequence of moves, $\left(y_{0},x_{1},y_{2},x_{3},y_{4},\dots\right)$, may be rewritten as $\left(y_{0},\left(x_{1},y_{2}\right) ,\left(x_{3},y_{4}\right) ,\dots\right)$, consisting of an initial move by $Y$ followed by a sequence of two-round pairs in which $X$ moves first and $Y$ moves second. The average payoff to player $X$ for this sequence of play is then
\begin{linenomath}
\begin{align}
\pi_{X} &= \left(1-\lambda\right)\left[\sum_{t=0}^{\infty}\lambda^{2t+1}f_{X}\left(x_{2t+1}\right) + \sum_{t=0}^{\infty}\lambda^{2t}g_{X}\left(y_{2t}\right)\right] \nonumber \\
&= \left(1-\lambda\right)\left[g_{X}\left(y_{0}\right) + \sum_{t=0}^{\infty}\lambda^{2t+1}f_{X}\left(x_{2t+1}\right) + \sum_{t=0}^{\infty}\lambda^{2t+2}g_{X}\left(y_{2t+2}\right)\right] \nonumber \\
&= \left(1-\lambda\right)\left[g_{X}\left(y_{0}\right) + \sum_{t=0}^{\infty}\lambda^{2t+1}\Big( f_{X}\left(x_{2t+1}\right) + \lambda g_{X}\left(y_{2t+2}\right)\Big) \right] \nonumber \\
&= \left(1-\lambda\right)\left[g_{X}\left(y_{0}\right) + \sum_{t=0}^{\infty}\lambda^{2t+1}u_{X}\left(x_{2t+1},y_{2t+2}\right) \right] .
\end{align}
\end{linenomath}
Similarly, player $Y$ has an average payoff of $\pi_{Y}=\left(1-\lambda\right)\left[g_{Y}\left(y_{0}\right) + \sum_{t=0}^{\infty}\lambda^{2t+1}u_{Y}\left(x_{2t+1},y_{2t+2}\right) \right]$.

The set of time-$T$ histories is then given by $\mathcal{H}^{T}:=\prod_{t=0}^{T-1}\mathcal{H}_{t}^{T}$, where
\begin{linenomath}
\begin{align}
\mathcal{H}_{t}^{T} &:= \begin{cases}S_{X} & t\textrm{ is odd}, \\ S_{Y} & t\textrm{ is even}\end{cases}
\end{align}
\end{linenomath}
for $0\leqslant t\leqslant T-1$, i.e. obtained from Eq. (\ref{sieq:historyStrictI}) by swapping $S_{X}$ and $S_{Y}$. Similarly, behavioral strategies for players $X$ and $Y$, respectively, are defined as maps,
\begin{linenomath}
\begin{subequations}
\begin{align}
\sigma_{X} &: \bigsqcup_{T\geqslant 0}\mathcal{H}^{2T+1} \longrightarrow \Delta\left(S_{X}\right) ; \\
\sigma_{Y} &: \bigsqcup_{T\geqslant 0}\mathcal{H}^{2T} \longrightarrow \Delta\left(S_{Y}\right) ,
\end{align}
\end{subequations}
\end{linenomath}
where, again, $\mathcal{H}^{0}:=\left\{\varnothing\right\}$ denotes the ``empty" history. In this case, we define
\begin{linenomath}
\begin{align}
\sigma &: \mathcal{H} := \bigsqcup_{T\geqslant 0}\mathcal{H}^{T} \longrightarrow \Delta\left(S_{X}\right)\sqcup\Delta\left(S_{Y}\right) \nonumber \\
&: h^{T} \longmapsto \begin{cases}\sigma_{X}\left[h^{T}\right] & T\textrm{ is odd}, \\ \sigma_{Y}\left[h^{T}\right] & T\textrm{ is even}.\end{cases}
\end{align}
\end{linenomath}
In terms of $\sigma$, the measures $\left\{\mu_{t}\right\}_{t\geqslant 0}$ and $\left\{\nu_{t}^{k}\right\}_{t\geqslant 0}^{0\leqslant k\leqslant t}$ are defined in the same way as they were in \S\ref{si:subsubsec:XmovesFirst}.

In a $\left(2T+1\right)$-round game (rounds $0$ through $2T$), the expected payoff to player $X$ is
\begin{linenomath}
\begin{align}
\pi_{X}^{2T+1} &:= \int\limits_{h^{2T+1}\in\mathcal{H}^{2T+1}} \left[ \left(\frac{1-\lambda}{1-\lambda^{2T}}\right)\left(g_{X}\left(h_{0}^{2T+1}\right) + \sum_{t=0}^{T-1}\lambda^{2t+1}u_{X}\left(h_{2t+1}^{2T+1},h_{2t+2}^{2T+1}\right)\right) \right] \nonumber \\
&\quad\quad\quad\quad \,d\sigma\left(h_{\leqslant 2T-1}^{2T+1},h_{2T}^{2T+1}\right)\cdots\,d\sigma\left(h_{\leqslant 0}^{2T+1},h_{1}^{2T+1}\right)\,d\sigma\left(\varnothing, h_{0}^{2T+1}\right) \nonumber \\
&= \left(\frac{1-\lambda}{1-\lambda^{2T}}\right)\int\limits_{h^{2T+1}\in\mathcal{H}^{2T+1}} g_{X}\left(h_{0}^{2T+1}\right) \,d\sigma\left(h_{\leqslant 2T-1}^{2T+1},h_{2T}^{2T+1}\right)\cdots\,d\sigma\left(h_{\leqslant 0}^{2T+1},h_{1}^{2T+1}\right)\,d\sigma\left(\varnothing, h_{0}^{2T+1}\right) \nonumber \\
&\quad +\left(\frac{1-\lambda}{1-\lambda^{2T}}\right)\sum_{t=0}^{T-1}\lambda^{2t+1}\int\limits_{h^{2T+1}\in\mathcal{H}^{2T+1}} u_{X}\left(h_{2t+1}^{2T+1},h_{2t+2}^{2T+1}\right) \nonumber \\
&\quad\quad\quad\quad \,d\sigma\left(h_{\leqslant 2T-1}^{2T+1},h_{2T}^{2T+1}\right)\cdots\,d\sigma\left(h_{\leqslant 0}^{2T+1},h_{1}^{2T+1}\right)\,d\sigma\left(\varnothing, h_{0}^{2T+1}\right) \nonumber \\\
&= \left(\frac{1-\lambda}{1-\lambda^{2T}}\right)\int\limits_{h^{1}\in\mathcal{H}^{1}} g_{X}\left(h_{0}^{1}\right) \,d\sigma\left(\varnothing, h_{0}^{1}\right) \nonumber \\
&\quad +\left(\frac{1-\lambda}{1-\lambda^{2T}}\right)\sum_{t=0}^{T-1}\lambda^{2t+1}\int\limits_{\left(h_{2t+1}^{2t+3},h_{2t+2}^{2t+3}\right)\in\mathcal{H}_{2t+1}^{2t+3}\times\mathcal{H}_{2t+2}^{2t+3}} u_{X}\left(h_{2t+1}^{2t+3},h_{2t+2}^{2t+3}\right) \,d\nu_{2t+2}^{1}\left(h_{2t+1}^{2t+3},h_{2t+2}^{2t+3}\right) \nonumber \\
&= \left(\frac{1-\lambda}{1-\lambda^{2T}}\right)\int\limits_{y_{0}\in S_{Y}} g_{X}\left(y_{0}\right) \,d\sigma_{Y}^{0}\left(y_{0}\right) \nonumber \\
&\quad +\left(\frac{1-\lambda}{1-\lambda^{2T}}\right)\sum_{t=0}^{T-1}\lambda^{2t+1}\int\limits_{\left(x,y\right)\in S_{X}\times S_{Y}} u_{X}\left(x,y\right) \,d\nu_{2t+2}^{1}\left(x,y\right) ,
\end{align}
\end{linenomath}
where $\sigma_{Y}^{0}:=\sigma_{Y}\left[\varnothing\right]$ is the initial action of player $Y$. Thus, we define player $X$'s average payoff as
\begin{linenomath}
\begin{align}\label{sieq:objectiveXStrictII}
\pi_{X} &:= \lim_{T\rightarrow\infty}\pi_{X}^{2T+1} \nonumber \\
&= \left(1-\lambda\right)\left[\int\limits_{y_{0}\in S_{Y}} g_{X}\left(y_{0}\right) \,d\sigma_{Y}^{0}\left(y_{0}\right) +\sum_{t=0}^{\infty}\lambda^{2t+1}\int\limits_{\left(x,y\right)\in S_{X}\times S_{Y}} u_{X}\left(x,y\right) \,d\nu_{2t+2}^{1}\left(x,y\right)\right] .
\end{align}
\end{linenomath}
Similarly, the expected payoff to player $Y$ is
\begin{linenomath}
\begin{align}\label{sieq:objectiveYStrictII}
\pi_{Y} &:= \left(1-\lambda\right)\left[\int\limits_{y_{0}\in S_{Y}} g_{Y}\left(y_{0}\right) \,d\sigma_{Y}^{0}\left(y_{0}\right) +\sum_{t=0}^{\infty}\lambda^{2t+1}\int\limits_{\left(x,y\right)\in S_{X}\times S_{Y}} u_{Y}\left(x,y\right) \,d\nu_{2t+2}^{1}\left(x,y\right)\right] .
\end{align}
\end{linenomath}

Once again, our main technical lemma is an analogue of Lemma 3.1 of \citet{akin:Games:2015}:
\begin{lemma}\label{lem:mainLemmaStrictII}
For any memory-one strategy, $\sigma_{X}\left[x,y\right]$, and any $E\in\mathcal{F}\left(S_{X}\right)$,
\begin{linenomath}
\begin{align}
\sum_{t=0}^{\infty}\lambda^{2t+1}\int\limits_{\left(x,y\right)\in S_{X}\times S_{Y}}\Big[ \chi_{E\times S_{Y}}\left(x,y\right) - \lambda^{2}\sigma_{X}\left[x,y\right]\left(E\right)\Big]\,d\nu_{2t+2}^{1}\left(x,y\right) &= \lambda\int\limits_{y_{0}\in S_{Y}}\sigma_{X}\left[y_{0}\right]\left(E\right)\,d\sigma_{Y}^{0}\left(y_{0}\right) ,
\end{align}
\end{linenomath}
where $\sigma_{X}^{0}\left[y_{0}\right]$ is the initial action of player $X$.
\end{lemma}
\begin{proof}
By the definition of $\left\{\nu_{t}^{k}\right\}_{t\geqslant 0}^{0\leqslant k\leqslant t}$, we see that
\begin{linenomath}
\begin{subequations}
\begin{align}
\int\limits_{\left(x,y\right)\in S_{X}\times S_{Y}}\chi_{E\times S_{Y}}\left(x,y\right)\,d\nu_{2t+2}^{1}\left(x,y\right) &= \nu_{2t+1}^{0}\left(E\right) ; \\
\int\limits_{\left(x,y\right)\in S_{X}\times S_{Y}}\sigma_{X}\left[x,y\right]\left(E\right)\,d\nu_{2t+2}^{1}\left(x,y\right) &= \nu_{2t+3}^{0}\left(E\right) .\end{align}
\end{subequations}
\end{linenomath}
Therefore, it follows that
\begin{linenomath}
\begin{align}
\sum_{t=0}^{\infty}&\lambda^{2t+1}\int\limits_{\left(x,y\right)\in S_{X}\times S_{Y}}\Big[ \chi_{E\times S_{Y}}\left(x,y\right) - \lambda^{2}\sigma_{X}\left[x,y\right]\left(E\right)\Big]\,d\nu_{2t+2}^{1}\left(x,y\right) \nonumber \\
&= \sum_{t=0}^{\infty}\lambda^{2t+1}\Big(\nu_{2t+1}^{0}\left(E\right) -\lambda^{2}\nu_{2t+3}^{0}\left(E\right)\Big) \nonumber \\
&= \lambda\nu_{1}^{0}\left(E\right) - \lim_{t\rightarrow\infty}\lambda^{2t+3}\nu_{2t+3}^{0}\left(E\right) \nonumber \\
&= \lambda\nu_{1}^{0}\left(E\right) \nonumber \\
&= \lambda\int\limits_{y_{0}\in S_{Y}}\sigma_{X}\left[y_{0}\right]\left(E\right)\,d\sigma_{Y}^{0}\left(y_{0}\right) ,
\end{align}
\end{linenomath}
which completes the proof.
\end{proof}

\begin{proposition}\label{prop:mainPropStrictII}
For any bounded, measurable function, $\psi :S_{X}\rightarrow\mathbb{R}$,
\begin{linenomath}
\begin{align}
\sum_{t=0}^{\infty} &\lambda^{2t+1}\int\limits_{\left(x,y\right)\in S_{X}\times S_{Y}} \left[ \psi\left(x\right) - \lambda^{2}\int\limits_{s\in S_{X}}\psi\left(s\right)\,d\sigma_{X}\left[x,y\right]\left(s\right) \right] \,d\nu_{2t+2}^{1}\left(x,y\right) \nonumber \\
&= \lambda\int\limits_{y_{0}\in S_{Y}}\int\limits_{s\in S_{X}}\psi\left(s\right)\,d\sigma_{X}^{0}\left[y_{0}\right]\left(s\right)\,d\sigma_{Y}^{0}\left(y_{0}\right) .
\end{align}
\end{linenomath}
\end{proposition}
\begin{proof}
The result follows from Lemma \ref{lem:mainLemmaStrictII} and the dominated convergence theorem. We do not include the details here; the argument is the same as the proof of Proposition 1 of \citep{mcavoy:PNAS:2016}.
\end{proof}

We are now in a position to prove the second of our main results for strictly-alternating games:
\begin{strictYTheorem}[Autocratic strategies for strictly-alternating games in which $Y$ moves first]
Suppose that
\begin{linenomath}
\begin{align}\label{sieq:conditionStrictII}
\Big( \alpha &f_{X}\left(x\right) + \beta f_{Y}\left(x\right) + \gamma \Big) + \lambda\Big( \alpha g_{X}\left(y\right) + \beta g_{Y}\left(y\right) + \gamma \Big) + \left(\frac{1-\lambda^{2}}{\lambda}\right)\Big( \alpha g_{X}\left(y_{0}\right) + \beta g_{Y}\left(y_{0}\right) + \gamma \Big) \nonumber \\
&= \psi\left(x\right) - \lambda^{2}\int\limits_{s\in S_{X}}\psi\left(s\right)\,d\sigma_{X}\left[x,y\right]\left(s\right) - \left(1-\lambda^{2}\right)\int\limits_{s\in S_{X}}\psi\left(s\right)\,d\sigma_{X}^{0}\left[y_{0}\right]\left(s\right)
\end{align}
\end{linenomath}
holds for some bounded $\psi$ and for each $x\in S_{X}$ and $y_{0},y\in S_{Y}$. Then, if player $X$ moves second, the pair $\left(\sigma_{X}^{0}\left[y_{0}\right],\sigma_{X}\left[x,y\right]\right)$ allows $X$ to enforce the equation $\alpha\pi_{X}+\beta\pi_{Y}+\gamma =0$ for every strategy of player $Y$.
\end{strictYTheorem}
\begin{proof}
If Eq. (\ref{sieq:conditionStrictII}) holds, then, for any initial action of player $Y$, $\sigma_{Y}^{0}$, we have
\begin{linenomath}
\begin{align}
\Big( \alpha &f_{X}\left(x\right) + \beta f_{Y}\left(x\right) + \gamma \Big) + \lambda\Big( \alpha g_{X}\left(y\right) + \beta g_{Y}\left(y\right) + \gamma \Big) \nonumber \\
&\quad + \left(\frac{1-\lambda^{2}}{\lambda}\right)\int\limits_{y_{0}\in S_{Y}}\Big[ \alpha g_{X}\left(y_{0}\right) + \beta g_{Y}\left(y_{0}\right) + \gamma \Big]\,d\sigma_{Y}^{0}\left(y_{0}\right) \nonumber \\
&= \psi\left(x\right) - \lambda^{2}\int\limits_{s\in S_{X}}\psi\left(s\right)\,d\sigma_{X}\left[x,y\right]\left(s\right) - \left(1-\lambda^{2}\right)\int\limits_{y_{0}\in S_{Y}}\int\limits_{s\in S_{X}}\psi\left(s\right)\,d\sigma_{X}^{0}\left[y_{0}\right]\left(s\right)\,d\sigma_{Y}^{0}\left(y_{0}\right) .
\end{align}
\end{linenomath}
Therefore, by Proposition \ref{prop:mainPropStrictII} and Eqs. (\ref{sieq:objectiveXStrictII}) and (\ref{sieq:objectiveYStrictII}), we see that, for each $\sigma_{Y}^{0}$,
\begin{linenomath}
\begin{align}
\alpha\pi_{X} + &\beta\pi_{Y} + \gamma +\left(1-\lambda\right)\lambda\int\limits_{y_{0}\in S_{Y}}\int\limits_{s\in S_{X}}\psi\left(s\right)\,d\sigma_{X}^{0}\left[y_{0}\right]\left(s\right)\,d\sigma_{Y}^{0}\left(y_{0}\right) \nonumber \\
&\quad -\left(1-\lambda\right)\int\limits_{y_{0}\in S_{Y}}\Big[\alpha g_{X}\left(y_{0}\right) +\beta g_{Y}\left(y_{0}\right) +\gamma\Big]\,d\sigma_{Y}^{0}\left(y_{0}\right) \nonumber \\
&= \left(1-\lambda\right)\sum_{t=0}^{\infty}\lambda^{2t+1}\int\limits_{\left(x,y\right)\in S_{X}\times S_{Y}} \left[ \psi\left(x\right) - \lambda^{2}\int\limits_{s\in S_{X}}\psi\left(s\right)\,d\sigma_{X}\left[x,y\right]\left(s\right) \right] \,d\nu_{2t+2}^{1}\left(x,y\right) \nonumber \\
&\quad -\left(1-\lambda\right)\sum_{t=0}^{\infty}\lambda^{2t+1}\left( \left(\frac{1-\lambda^{2}}{\lambda}\right)\int\limits_{y_{0}\in S_{Y}}\Big[ \alpha g_{X}\left(y_{0}\right) + \beta g_{Y}\left(y_{0}\right) + \gamma \Big]\,d\sigma_{Y}^{0}\left(y_{0}\right) \right) \nonumber \\
&= \left(1-\lambda\right)\lambda\int\limits_{y_{0}\in S_{Y}}\int\limits_{s\in S_{X}}\psi\left(s\right)\,d\sigma_{X}^{0}\left[y_{0}\right]\left(s\right)\,d\sigma_{Y}^{0}\left(y_{0}\right) \nonumber \\
&\quad - \left(1-\lambda\right)\int\limits_{y_{0}\in S_{Y}}\Big[ \alpha g_{X}\left(y_{0}\right) + \beta g_{Y}\left(y_{0}\right) + \gamma \Big]\,d\sigma_{Y}^{0}\left(y_{0}\right) ,
\end{align}
\end{linenomath}
and it follows immediately that $\alpha\pi_{X}+\beta\pi_{Y}+\gamma =0$.
\end{proof}

\section{Randomly-alternating games}\label{si:subsec:random}

In each round of a randomly-alternating game, player $X$ moves with probability $\omega_{X}$ and player $Y$ moves with probability $1-\omega_{X}$ for some $0\leqslant\omega_{X}\leqslant 1$. For $T\geqslant 1$, a time-$T$ history is an element of the space
\begin{linenomath}
\begin{align}
\mathcal{H}^{T} &:= \left(S_{X}\sqcup S_{Y}\right)^{T} ,
\end{align}
\end{linenomath}
where $S_{X}\sqcup S_{Y}$ denotes the disjoint union of the action spaces of the players, $S_{X}$ and $S_{Y}$. As in \S\ref{si:subsec:strict}, we let $\mathcal{H}^{0}:=\left\{\varnothing\right\}$, where $\varnothing$ indicates the ``empty history." In terms of the space of all possible histories, $\mathcal{H}:=\left\{\varnothing\right\}\sqcup\bigsqcup_{T\geqslant 1}\mathcal{H}^{T}$, behavioral strategies for players $X$ and $Y$, respectively, are maps,
\begin{linenomath}
\begin{subequations}
\begin{align}
\sigma_{X} &: \mathcal{H} \longrightarrow \Delta\left(S_{X}\right) ; \\
\sigma_{Y} &: \mathcal{H} \longrightarrow \Delta\left(S_{Y}\right) .
\end{align}
\end{subequations}
\end{linenomath}
These strategies may be written more compactly as a single map, $\sigma :\mathcal{H}\rightarrow\Delta\left(S_{X}\sqcup S_{Y}\right)$, defined for $h^{T}\in\mathcal{H}^{T}$ and $E\in\mathcal{F}\left(S_{X}\sqcup S_{Y}\right)$ via $\sigma\left[h^{T}\right]\left(E\right) :=\omega_{X}\sigma_{X}\left[h^{T}\right]\left(E\cap S_{X}\right) +\left(1-\omega_{X}\right)\sigma_{Y}\left[h^{T}\right]\left(E\cap S_{Y}\right)$. Furthermore, if $\mathcal{H}_{t}^{T}:=S_{X}\sqcup S_{Y}$ for $0\leqslant t\leqslant T-1$, then these two strategies, $\sigma_{X}$ and $\sigma_{Y}$, together generate a sequence of probability measures, $\left\{\nu_{t}^{0}\right\}_{t\geqslant 0}$, on $\mathcal{H}_{t}^{t+1}=S_{X}\sqcup S_{Y}$ for each $t$, defined via Eqs. (\ref{sieq:muDef}) and (\ref{sieq:nuDef}).

Consider the single-round payoff function for player $X$, $u_{X}:S_{X}\sqcup S_{Y}\rightarrow\mathbb{R}$, defined by
\begin{linenomath}
\begin{align}\label{sieq:singleRoundRandom}
u_{X}\left(s\right) &:= \begin{cases}f_{X}\left(s\right) & s\in S_{X} , \\ g_{X}\left(s\right) & s\in S_{Y} .\end{cases}
\end{align}
\end{linenomath}
The single-round payoff function for player $Y$, $u_{Y}$, is defined by replacing $f_{X}$ by $f_{Y}$ and $g_{X}$ by $g_{Y}$ in Eq. (\ref{sieq:singleRoundRandom}). In a $\left(T+1\right)$-round game (rounds $0$ through $T$), the expected payoff to player $X$ is then
\begin{linenomath}
\begin{align}
\pi_{X}^{T+1} &:= \int\limits_{h^{T+1}\in\mathcal{H}^{T+1}} \left[ \left(\frac{1-\lambda}{1-\lambda^{T+1}}\right)\sum_{t=0}^{T}\lambda^{t}u_{X}\left(h_{t}^{T+1}\right) \right] \,d\sigma\left(h_{\leqslant T-1}^{T+1},h_{T}^{T+1}\right)\cdots\,d\sigma\left(h_{\leqslant 0}^{T+1},h_{1}^{T+1}\right)\,d\sigma\left(\varnothing, h_{0}^{T+1}\right) \nonumber \\
&= \left(\frac{1-\lambda}{1-\lambda^{T+1}}\right)\sum_{t=0}^{T}\lambda^{t}\int\limits_{h^{T+1}\in\mathcal{H}^{T+1}} u_{X}\left(h_{t}^{T+1}\right) \,d\sigma\left(h_{\leqslant T-1}^{T+1},h_{T}^{T+1}\right)\cdots\,d\sigma\left(h_{\leqslant 0}^{T+1},h_{1}^{T+1}\right)\,d\sigma\left(\varnothing, h_{0}^{T+1}\right) \nonumber \\
&= \left(\frac{1-\lambda}{1-\lambda^{T+1}}\right)\sum_{t=0}^{T}\lambda^{t}\int\limits_{h^{t+1}\in\mathcal{H}^{t+1}} u_{X}\left(h_{t}^{t+1}\right) \,d\sigma\left(h_{\leqslant t-1}^{t+1},h_{t}^{t+1}\right)\cdots\,d\sigma\left(h_{\leqslant 0}^{t+1},h_{1}^{t+1}\right)\,d\sigma\left(\varnothing, h_{0}^{t+1}\right) \nonumber \\
&= \left(\frac{1-\lambda}{1-\lambda^{T+1}}\right)\sum_{t=0}^{T}\lambda^{t}\int\limits_{h_{t}^{t+1}\in\mathcal{H}_{t}^{t+1}} u_{X}\left(h_{t}^{t+1}\right) \, d\nu_{t}^{0}\left(h_{t}^{t+1}\right) \nonumber \\
&= \left(\frac{1-\lambda}{1-\lambda^{T+1}}\right)\sum_{t=0}^{T}\lambda^{t}\int\limits_{s\in S_{X}\sqcup S_{Y}} u_{X}\left(s\right) \, d\nu_{t}^{0}\left(s\right) .
\end{align}
\end{linenomath}
Therefore, we define the average payoff of player $X$ to be
\begin{linenomath}
\begin{align}\label{sieq:objectiveXRandom}
\pi_{X} &:= \lim_{T\rightarrow\infty}\pi_{X}^{T+1} = \left(1-\lambda\right)\sum_{t=0}^{\infty}\lambda^{t}\int\limits_{s\in S_{X}\sqcup S_{Y}} u_{X}\left(s\right) \, d\nu_{t}^{0}\left(s\right) .
\end{align}
\end{linenomath}
Similarly, the expected payoff of player $Y$ is
\begin{linenomath}
\begin{align}\label{sieq:objectiveYRandom}
\pi_{Y} &:= \left(1-\lambda\right)\sum_{t=0}^{\infty}\lambda^{t}\int\limits_{s\in S_{X}\sqcup S_{Y}} u_{Y}\left(s\right) \, d\nu_{t}^{0}\left(s\right) .
\end{align}
\end{linenomath}

A memory-one strategy in the context of randomly-alternating games looks slightly different from that of strictly-alternating games. Instead of simply reacting to the previous moves of the players, one also needs to know which player moved last since, in any given round, either $X$ or $Y$ could move (provided $\omega_{X}\neq 0,1$). Therefore, a memory-one strategy for player $X$ consists of an action policy, $\sigma_{X}^{X}\left[x\right]$, when $X$ moves $x$ in the previous round, and a policy, $\sigma_{X}^{Y}\left[y\right]$, when $Y$ moves $y$ in the previous round. More succinctly, we let
\begin{linenomath}
\begin{align}\label{sieq:XYcombinedRandom}
\sigma_{X}\left[s\right] &:= \begin{cases}\sigma_{X}^{X}\left[s\right] & s\in S_{X} , \\ \sigma_{X}^{Y}\left[s\right] & s\in S_{Y} .\end{cases}
\end{align}
\end{linenomath}

One final time, our main technical lemma is an analogue of Lemma 3.1 of \citet{akin:Games:2015}:
\begin{lemma}\label{lem:mainLemmaRandom}
For memory-one strategies, $\sigma_{X}^{X}\left[x\right]$ and $\sigma_{X}^{Y}\left[y\right]$, and $E\in\mathcal{F}\left(S_{X}\right)$, we have
\begin{linenomath}
\begin{align}
\sum_{t=0}^{\infty}\lambda^{t}\int\limits_{s\in S_{X}\sqcup S_{Y}}\Big[\chi_{E}\left(s\right) - \lambda\omega_{X}\sigma_{X}\left[s\right]\left(E\right)\Big]\,d\nu_{t}^{0}\left(s\right) &= \omega_{X}\sigma_{X}^{0}\left(E\right) ,
\end{align}
\end{linenomath}
where $\sigma_{X}\left[s\right]$ is defined via Eq. (\ref{sieq:XYcombinedRandom}).
\end{lemma}
\begin{proof}
By the definition of the sequence of measures, $\left\{\nu_{t}^{0}\right\}_{t\geqslant 0}$,
\begin{linenomath}
\begin{subequations}
\begin{align}
\int\limits_{s\in S_{X}\sqcup S_{Y}}\chi_{E}\left(s\right)\,d\nu_{t}^{0}\left(s\right) &= \nu_{t}^{0}\left(E\right) ; \\
\int\limits_{s\in S_{X}\sqcup S_{Y}}\omega_{X}\sigma_{X}\left[s\right]\left(E\right)\,d\nu_{t}^{0}\left(s\right) &= \nu_{t+1}^{0}\left(E\right) .
\end{align}
\end{subequations}
\end{linenomath}
Therefore, we see that
\begin{linenomath}
\begin{align}
\sum_{t=0}^{\infty} &\lambda^{t}\int\limits_{s\in S_{X}\sqcup S_{Y}}\Big[\chi_{E}\left(s\right) - \lambda\omega_{X}\sigma_{X}\left[s\right]\left(E\right)\Big]\,d\nu_{t}^{0}\left(s\right) \nonumber \\
&= \sum_{t=0}^{\infty}\lambda^{t}\Big( \nu_{t}^{0}\left(E\right) - \lambda\nu_{t+1}^{0}\left(E\right) \Big) \nonumber \\
&= \nu_{0}^{0}\left(E\right) - \lim_{t\rightarrow\infty}\lambda^{t+1}\nu_{t+1}^{0}\left(E\right) \nonumber \\
&= \nu_{0}^{0}\left(E\right) \nonumber \\
&= \omega_{X}\sigma_{X}^{0}\left(E\right) ,
\end{align}
\end{linenomath}
which completes the proof.
\end{proof}

\begin{proposition}\label{prop:mainPropRandom}
For any bounded, measurable function, $\psi :S_{X}\sqcup S_{Y}\rightarrow\mathbb{R}$, with $\textrm{supp}\,\psi\subseteq S_{X}$,
\begin{linenomath}
\begin{align}
\sum_{t=0}^{\infty}\lambda^{t}\int\limits_{s\in S_{X}\sqcup S_{Y}}\left[\psi\left(s\right) - \lambda\omega_{X}\int\limits_{s'\in S_{X}}\psi\left(s'\right)\,\sigma_{X}\left[s\right]\left(s'\right)\right]\,d\nu_{t}^{0}\left(s\right) &= \omega_{X}\int\limits_{s\in S_{X}}\psi\left(s\right)\,d\sigma_{X}^{0}\left(s\right) .
\end{align}
\end{linenomath}
\end{proposition}
\begin{proof}
The result follows from Lemma \ref{lem:mainLemmaRandom} and the dominated convergence theorem. We do not include the details here; the argument is the same as the proof of Proposition 1 of \citep{mcavoy:PNAS:2016}.
\end{proof}

Proposition \ref{prop:mainPropRandom} allows us to prove our main result for randomly-alternating games:
\begin{randomTheorem}[Autocratic strategies for randomly-alternating games]
If, for some bounded $\psi$,
\begin{linenomath}
\begin{subequations}\label{sieq:randomConds}
\begin{align}
\alpha f_{X}\left(x\right) + \beta f_{Y}\left(x\right) + \gamma &= \psi\left(x\right) - \lambda \omega_{X}\int\limits_{s\in S_{X}}\psi\left(s\right)\,d\sigma_{X}^{X}\left[x\right]\left(s\right) -\left(1-\lambda\right)\omega_{X}\int\limits_{s\in S_{X}}\psi\left(s\right)\,d\sigma_{X}^{0}\left(s\right) ; \label{sieq:randomCondI} \\
\alpha g_{X}\left(y\right) + \beta g_{Y}\left(y\right) + \gamma &= -\lambda \omega_{X}\int\limits_{s\in S_{X}}\psi\left(s\right)\,\sigma_{X}^{Y}\left[y\right]\left(s\right) -\left(1-\lambda\right)\omega_{X}\int\limits_{s\in S_{X}}\psi\left(s\right)\,d\sigma_{X}^{0}\left(s\right) \label{sieq:randomCondII}
\end{align}
\end{subequations}
\end{linenomath}
for each $x\in S_{X}$ and $y\in S_{Y}$, then the strategy $\left(\sigma_{X}^{0},\sigma_{X}^{X}\left[x\right] ,\sigma_{X}^{Y}\left[y\right]\right)$ allows $X$ to enforce the equation $\alpha\pi_{X}+\beta\pi_{Y}+\gamma =0$ for every strategy of player $Y$.
\end{randomTheorem}
\begin{proof}
If Eq. (\ref{sieq:randomConds}) holds, then, by Proposition \ref{prop:mainPropRandom} and Eqs. (\ref{sieq:objectiveXRandom}) and (\ref{sieq:objectiveYRandom}),
\begin{linenomath}
\begin{align}
\alpha\pi_{X} + &\beta\pi_{Y} + \gamma + \left(1-\lambda\right)\omega_{X}\int\limits_{s\in S_{X}}\psi\left(s\right)\,d\sigma_{X}^{0}\left(s\right) \nonumber \\
&= \left(1-\lambda\right)\sum_{t=0}^{\infty}\lambda^{t}\int\limits_{s\in S_{X}\sqcup S_{Y}}\left[\psi\left(s\right) - \lambda\omega_{X}\int\limits_{s'\in S_{X}}\psi\left(s'\right)\,\sigma_{X}\left[s\right]\left(s'\right)\right]\,d\nu_{t}^{0}\left(s\right) \nonumber \\
&= \left(1-\lambda\right)\omega_{X}\int\limits_{s\in S_{X}}\psi\left(s\right)\,d\sigma_{X}^{0}\left(s\right) ,
\end{align}
\end{linenomath}
from which it follows that $\alpha\pi_{X}+\beta\pi_{Y}+\gamma =0$, as desired.
\end{proof}

\subsection{Two-point autocratic strategies}\label{si:subsec:twoPointStrategies}

Suppose that $X$ wishes to enforce $\alpha\pi_{X}+\beta\pi_{Y}+\gamma =0$ with
\begin{linenomath}
\begin{subequations}\label{sieq:twoPointRandomDef}
\begin{align}
\sigma_{X}^{0} &= p_{0}\delta_{s_{1}}+\left(1-p_{0}\right)\delta_{s_{2}} ; \\
\sigma_{X}^{X}\left[x\right] &= p^{X}\left(x\right)\delta_{s_{1}}+\Big(1-p^{X}\left(x\right)\Big)\delta_{s_{2}} ; \\
\sigma_{X}^{Y}\left[y\right] &= p^{Y}\left(y\right)\delta_{s_{1}}+\Big(1-p^{Y}\left(y\right)\Big)\delta_{s_{2}}
\end{align}
\end{subequations}
\end{linenomath}
for some $s_{1}$ and $s_{2}$ in $S_{X}$. Consider the function, $\varphi :S_{X}\sqcup S_{Y}\rightarrow\mathbb{R}$, defined by
\begin{linenomath}
\begin{align}
\varphi\left(s\right) &:= \begin{cases}\alpha f_{X}\left(s\right) + \beta f_{Y}\left(s\right) + \gamma & s\in S_{X} ; \\ \alpha g_{X}\left(s\right) + \beta g_{Y}\left(s\right) + \gamma & s\in S_{Y} .\end{cases}
\end{align}
\end{linenomath}
Then, in terms of $\varphi$, it must be the case that
\begin{linenomath}
\begin{subequations}
\begin{align}
p^{X}\left(x\right) &= \frac{\frac{1}{\lambda\omega_{X}}\left(\psi\left(x\right) - \varphi\left(x\right) - \left(1-\lambda\right)\omega_{X}\Big(\psi\left(s_{1}\right) p_{0}+\psi\left(s_{2}\right)\left(1-p_{0}\right)\Big)\right) - \psi\left(s_{2}\right)}{\psi\left(s_{1}\right) - \psi\left(s_{2}\right)} ; \label{sieq:twoPointRandomI} \\
p^{Y}\left(y\right) &= \frac{\frac{1}{\lambda\omega_{X}}\left(- \varphi\left(y\right) - \left(1-\lambda\right)\omega_{X}\Big(\psi\left(s_{1}\right) p_{0}+\psi\left(s_{2}\right)\left(1-p_{0}\right)\Big)\right) - \psi\left(s_{2}\right)}{\psi\left(s_{1}\right) - \psi\left(s_{2}\right)} . \label{sieq:twoPointRandomII}
\end{align}
\end{subequations}
Therefore, provided $0\leqslant p^{X}\left(x\right)\leqslant 1$ and $0\leqslant p^{Y}\left(y\right)\leqslant 1$ for each $x\in S_{X}$ and $y\in S_{Y}$, the two-point strategy defined by Eq. (\ref{sieq:twoPointRandomDef}) allows player $X$ to unilaterally enforce the relationship $\alpha\pi_{X}+\beta\pi_{Y}+\gamma =0$.
\end{linenomath}

\subsection{Deterministic autocratic strategies}\label{si:subsec:deterministic}

Suppose that $X$ wishes to enforce $\alpha\pi_{X}+\beta\pi_{Y}+\gamma =0$ by using a deterministic strategy, which is defined in terms of a reaction function to the previous move of the game. That is, a deterministic memory-one strategy for player $X$ consists of an initial action, $x_{0}\in S_{X}$, and two reaction functions, $r^{X}:S_{X}\rightarrow S_{X}$ and $r^{Y}:S_{Y}\rightarrow S_{X}$. Player $X$ begins by using $x_{0}$ with certainty. If $X$ uses $x$ in the previous round and $X$ moves again, then $X$ plays $r^{X}\left(x\right)$ in the subsequent round. On the other hand, if $Y$ moves $y$ in the previous round and $X$ follows this move, then $X$ plays $r^{Y}\left(y\right)$ in response to $Y$'s action. For a deterministic strategy with these reaction functions, Eq. (\ref{sieq:randomConds}) takes the form
\begin{linenomath}
\begin{subequations}
\begin{align}
\alpha f_{X}\left(x\right) + \beta f_{Y}\left(x\right) + \gamma &= \psi\left(x\right) - \lambda\omega_{X}\psi\left(r^{X}\left(x\right)\right) -\left(1-\lambda\right)\omega_{X}\psi\left(x_{0}\right) ; \\
\alpha g_{X}\left(y\right) + \beta g_{Y}\left(y\right) + \gamma &= -\lambda \omega_{X}\psi\left(r^{Y}\left(y\right)\right) -\left(1-\lambda\right)\omega_{X}\psi\left(x_{0}\right) .
\end{align}
\end{subequations}
\end{linenomath}

\section{Continuous Donation Game}\label{si:subsec:CDG}

The results we give for the continuous Donation Game hold for any benefit and cost functions, $b\left(s\right)$ and $c\left(s\right)$, and any interval of cooperation levels, $\left[0,K\right] =S_{X}=S_{Y}$. For the purposes of plotting feasible regions (Figs. \ref{fig:feasibleRegion} and \ref{fig:regionsStrictIandII}) and for performing simulations (Figs. \ref{fig:twoPointSimulation} and \ref{fig:deterministicSimulation}), we use for benefit and cost functions
\begin{linenomath}
\begin{subequations}
\begin{align}
b\left(s\right) &:= 5\left(1-e^{-2s}\right) ; \\
c\left(s\right) &:= 2s ,
\end{align}
\end{subequations}
\end{linenomath}
respectively, and these functions are defined on the interval $\left[0,2\right] =S_{X}=S_{Y}$ \citep[see][]{killingback:PRSB:1999,killingback:AN:2002}.

\subsection{Strictly-alternating moves}\label{si:subsubsec:CDGstrict}

Fig. \ref{fig:regionsStrictIandII} shows the feasible payoff regions for three values of $\lambda$ in the strictly-alternating game when $X$ or $Y$ moves first.
\begin{figure}
\begin{center}
\includegraphics[scale=0.55]{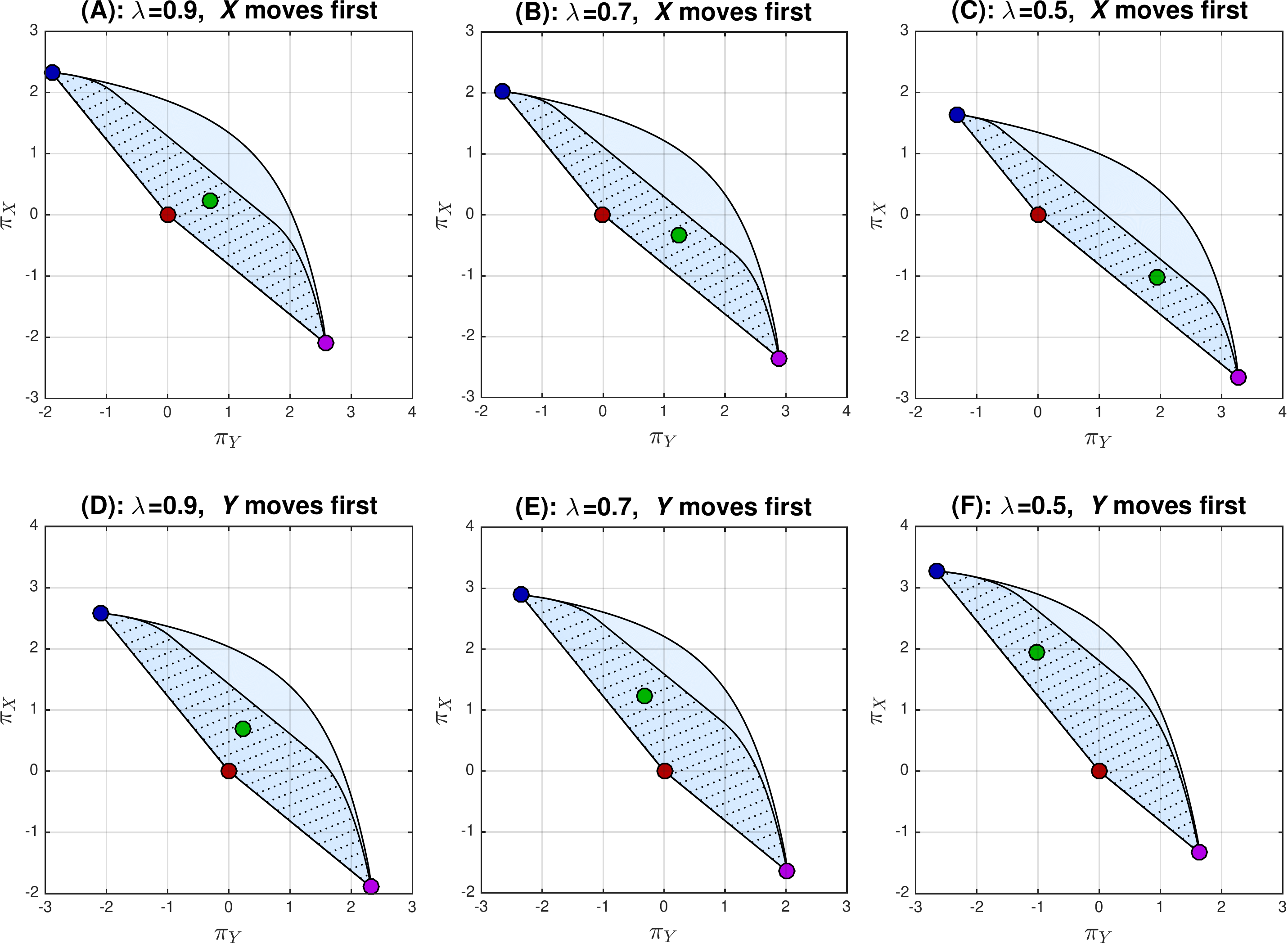}
\end{center}
\caption{Feasible payoff regions for three values of $\lambda$ in the strictly-alternating, continuous Donation Game when player $X$ moves first (top) and player $Y$ moves first (bottom). The shaded region represents the feasible payoffs when $X$ plays a two-point strategy (only $0$ and $K$). As the discounting factor, $\lambda$, gets smaller (i.e. discounting stronger), the first move has more of a pronounced effect on the expected payoffs.\label{fig:regionsStrictIandII}}
\end{figure}
Note that these regions depend on the discounting factor, $\lambda$, due to the payoff asymmetries introduced by the sequential moves even for symmetric interactions. In contrast, in the randomly-alternating, continuous Donation Game, these regions do not depend on $\lambda$.

\subsubsection{Initial actions}

In the main text, we presented two-point autocratic strategies that are concentrated on just $0$ and $K$. Here, we give the technical requirements of the probability of initially cooperating, $p_{0}$.

Player $X$ can unilaterally enforce $\pi_{X}-\kappa_{X}=\chi\left(\pi_{Y}-\kappa_{Y}\right)$ using the reaction probability
\begin{linenomath}
\begin{align}
p\left(y\right) &= \frac{\lambda\left(b\left(y\right) +\chi c\left(y\right)\right) + \left(1+\lambda\right)\left(\chi\kappa_{Y}-\kappa_{X}\right)}{\lambda^{2}\left(\chi b\left(K\right) +c\left(K\right)\right)} - \frac{1-\lambda^{2}}{\lambda^{2}} p_{0} ,
\end{align}
\end{linenomath}
provided her probability of cooperating in the first round, $p_{0}$, satisfies
\begin{linenomath}
\begin{align}\label{sieq:extGenStrictInitial}
\max &\left\{\frac{\!\lambda\left(b\left(K\right)\! +\!\chi c\left(K\right)\!\right) +\left(1+\lambda\right)\!\left(\chi\kappa_{Y}\!-\!\kappa_{X}\right)\!}{\left(1-\lambda^{2}\right)\!\left(\chi\, b\left(K\right) +c\left(K\right)\!\right)}-\frac{\lambda^2}{1-\lambda^2},0\right\} 
\leqslant p_{0} \leqslant \min\left\{\frac{\chi\kappa_{Y}-\kappa_{X}}{\left(1-\lambda\right)\!\left(\chi\, b\left(K\right)\! +c\left(K\right)\!\right)},1\right\} .
\end{align}
\end{linenomath}
Similarly, player $X$ can enforce $\pi_{Y}=\gamma$ using the reaction probability
\begin{linenomath}
\begin{align}
p\left(y\right) &= \frac{\lambda c\left(y\right) + \left(1+\lambda\right)\gamma}{\lambda^{2}b\left(K\right)} - \frac{1-\lambda^{2}}{\lambda^{2}} p_{0} ,
\end{align}
\end{linenomath}
provided $p_{0}$ satisfies
\begin{linenomath}
\begin{align}\label{sieq:equalizerStrictInitial}
\max\left\{\frac{\lambda c\left(K\right) +\left(1+\lambda\right)\gamma}{\left(1-\lambda^{2}\right) b\left(K\right)}-\frac{\lambda^2}{1-\lambda^2},0\right\} \leqslant p_{0} \leqslant \min\left\{\frac{\gamma}{\left(1-\lambda\right) b\left(K\right)},1\right\} .
\end{align}
\end{linenomath}

\subsection{Randomly-alternating moves}\label{si:subsubsec:CDGrandom}

\subsubsection{Extortionate and generous strategies}

Suppose that, via Eq. (\ref{sieq:randomConds}), $X$ can enforce $\pi_{X}=\chi\pi_{Y}-\gamma$ for some $\chi\geqslant 1$ and $\gamma\in\mathbb{R}$. Then, for some bounded function, $\psi :\left[0,K\right]\rightarrow\mathbb{R}$,
\begin{linenomath}
\begin{subequations}
\begin{align}
-c\left(x\right) - \chi b\left(x\right) + \gamma &= \psi\left(x\right) - \lambda \omega_{X}\int\limits_{s\in S_{X}}\psi\left(s\right)\,d\sigma_{X}^{X}\left[x\right]\left(s\right) -\left(1-\lambda\right)\omega_{X}\int\limits_{s\in S_{X}}\psi\left(s\right)\,d\sigma_{X}^{0}\left(s\right) ; \label{sieq:extGenFirst} \\
b\left(y\right) + \chi c\left(y\right) + \gamma &= -\lambda \omega_{X}\int\limits_{s\in S_{X}}\psi\left(s\right)\,\sigma_{X}^{Y}\left[y\right]\left(s\right) -\left(1-\lambda\right)\omega_{X}\int\limits_{s\in S_{X}}\psi\left(s\right)\,d\sigma_{X}^{0}\left(s\right) \label{sieq:extGenSecond}
\end{align}
\end{subequations}
\end{linenomath}
for each $x,y\in\left[0,K\right] =S_{X}=S_{Y}$. Eq. (\ref{sieq:extGenFirst}) implies that
\begin{linenomath}
\begin{align}
\left(1-\omega_{X}\right)\sup\psi \leqslant \gamma \leqslant \chi b\left(K\right) + c\left(K\right) + \left(1-\omega_{X}\right)\inf\psi ,
\end{align}
\end{linenomath}
and Eq. (\ref{sieq:extGenSecond}) implies that
\begin{linenomath}
\begin{align}
-\omega_{X}\sup\psi \leqslant \gamma \leqslant -b\left(K\right) -\chi c\left(K\right) -\omega_{X}\inf\psi .
\end{align}
\end{linenomath}
It follows at once from these inequalities that
\begin{linenomath}
\begin{align}
0 \leqslant \gamma \leqslant \chi\Big(\omega_{X}b\left(K\right) - \left(1-\omega_{X}\right) c\left(K\right)\Big) - \Big(\left(1-\omega_{X}\right) b\left(K\right) -\omega_{X}c\left(K\right)\Big) .
\end{align}
\end{linenomath}
In particular, if $\gamma =\chi\kappa_{Y}-\kappa_{X}$, then it must be true that
\begin{linenomath}
\begin{align}
\kappa_{X} = \kappa_{X} + \left( \chi\kappa_{Y}^{00} - \kappa_{X}^{00} \right) \leqslant \chi\kappa_{Y} \leqslant \kappa_{X} + \left( \chi\kappa_{Y}^{KK} - \kappa_{X}^{KK} \right) ,
\end{align}
\end{linenomath}
which is simply Eq. (\ref{eq:kappaInequalities}) in the main text.

\subsubsection{Equalizer strategies -- own score}\label{si:subsubsec:equalizers}

Player $X$ can ensure that $\pi_{X}=\gamma$ if $\sigma_{X}$ satisfies
\begin{linenomath}
\begin{subequations}\label{eq:setOwn}
\begin{align}
-c\left(x\right) - \gamma &= \psi\left(x\right) - \lambda \omega_{X}\int\limits_{s\in S_{X}}\psi\left(s\right)\,d\sigma_{X}^{X}\left[x\right]\left(s\right) -\left(1-\lambda\right)\omega_{X}\int\limits_{s\in S_{X}}\psi\left(s\right)\,d\sigma_{X}^{0}\left(s\right) ; \label{eq:setOwnFirst} \\
b\left(y\right) - \gamma &= -\lambda \omega_{X}\int\limits_{s\in S_{X}}\psi\left(s\right)\,\sigma_{X}^{Y}\left[y\right]\left(s\right) -\left(1-\lambda\right)\omega_{X}\int\limits_{s\in S_{X}}\psi\left(s\right)\,d\sigma_{X}^{0}\left(s\right) \label{eq:setOwnSecond}
\end{align}
\end{subequations}
\end{linenomath}
for some bounded function, $\psi$, and each $x,y\in\left[0,K\right]$. Eq. (\ref{eq:setOwnFirst}) implies that
\begin{linenomath}
\begin{align}
\gamma &\leqslant -\left(1-\omega_{X}\right)\sup\psi
\end{align}
\end{linenomath}
and Eq. (\ref{eq:setOwnSecond}) implies that
\begin{linenomath}
\begin{align}
\gamma &\leqslant \omega_{X}\sup\psi ,
\end{align}
\end{linenomath}
thus $\gamma\leqslant 0$. Therefore, player $X$ can unilaterally set her own score to at most $0$. However, it should be noted that, in contrast to the continuous Donation Game with simultaneous moves, it is possible for a player to set her own score (to at most $0$) when moves alternate randomly. For example, if $\gamma =0$ and $\psi\left(s\right) =-\frac{1}{\lambda\omega_{X}}b\left(s\right)$, then player $X$ can unilaterally set $\pi_{X}=0$ using
\begin{linenomath}
\begin{subequations}
\begin{align}
\sigma_{X}^{0} &= \delta_{0} ; \\
\sigma_{X}^{X}\left[x\right] &= \Big(1-p^{X}\left(x\right)\Big)\delta_{0} + p^{X}\left(x\right)\delta_{K} ; \\
\sigma_{X}^{Y}\left[y\right] &= \delta_{y} ,
\end{align}
\end{subequations}
\end{linenomath}
where $p^{X}\left(x\right) =\frac{1}{b\left(K\right)}\left(\frac{1}{\lambda\omega_{X}}b\left(x\right) -c\left(x\right)\right)$, provided $\lambda\omega_{X}$ is sufficiently close to $1$. Interestingly, however, if players move with equal probability, $\omega_{X}=1/2$, then player $X$ can never set her own score: Eq. (\ref{eq:setOwn}) implies that $\lambda\omega_{X}\geqslant\frac{b\left(K\right)}{b\left(K\right) +c\left(K\right)}>1/2$, which can never hold for $\omega_{X}=1/2$ and $0\leqslant\lambda\leqslant 1$. Even when a player can set her own score in the continuous Donation Game, this score can be at most $0$; thus, since a player can achieve at least $0$ by defecting in every round, such an equalizer strategy would never be attractive.

\subsubsection{Equalizer strategies -- opponent's score}\label{si:subsubsec:equalizersII}

We saw in \S\ref{subsubsec:randomCDG} that player $X$ can set $\pi_{Y}=\gamma$ for any $0\leqslant\gamma\leqslant\omega_{X}b\left(K\right) -\left(1-\omega_{X}\right) c\left(K\right)$. Here, we show that, using Eq. (\ref{eq:mainEquationRandom}), there are no other payoffs for player $Y$ that $X$ can set unilaterally. Indeed, suppose
\begin{linenomath}
\begin{subequations}
\begin{align}
b\left(x\right) - \gamma &= \psi\left(x\right) - \lambda \omega_{X}\int\limits_{s\in S_{X}}\psi\left(s\right)\,d\sigma_{X}^{X}\left[x\right]\left(s\right) -\left(1-\lambda\right)\omega_{X}\int\limits_{s\in S_{X}}\psi\left(s\right)\,d\sigma_{X}^{0}\left(s\right) ; \label{eq:setOpponentFirst} \\
-c\left(y\right) - \gamma &= -\lambda \omega_{X}\int\limits_{s\in S_{X}}\psi\left(s\right)\,\sigma_{X}^{Y}\left[y\right]\left(s\right) -\left(1-\lambda\right)\omega_{X}\int\limits_{s\in S_{X}}\psi\left(s\right)\,d\sigma_{X}^{0}\left(s\right) \label{eq:setOpponentSecond}
\end{align}
\end{subequations}
\end{linenomath}
for some bounded function, $\psi$, and each $x,y\in\left[0,K\right]$. From Eq. (\ref{eq:setOpponentFirst}), we see that
\begin{linenomath}
\begin{align}
-\left(1-\omega_{X}\right)\inf\psi &\leqslant \gamma \leqslant b\left(K\right) - \left(1-\omega_{X}\right)\sup\psi .
\end{align}
\end{linenomath}
Similarly, Eq. (\ref{eq:setOpponentSecond}) implies that
\begin{linenomath}
\begin{align}
\omega_{X}\inf\psi &\leqslant \gamma \leqslant -c\left(K\right) + \omega_{X}\sup\psi .
\end{align}
\end{linenomath}
These inequalities immediately give $0\leqslant\gamma\leqslant\min\Big\{b\left(K\right) - \left(1-\omega_{X}\right)\sup\psi , -c\left(K\right) + \omega_{X}\sup\psi\Big\}$. Since
\begin{linenomath}
\begin{align}
b\left(K\right) - \left(1-\omega_{X}\right)\sup\psi \leqslant -c\left(K\right) + \omega_{X}\sup\psi &\iff b\left(K\right) + c\left(K\right) \leqslant \sup\psi ,
\end{align}
\end{linenomath}
it follows that $0\leqslant\gamma\leqslant\omega_{X}b\left(K\right) - \left(1-\omega_{X}\right) c\left(K\right) =\kappa_{Y}^{KK}$.

\subsubsection{Initial actions}

Here, we give the technical conditions on $p_{0}$, the probability that $X$ cooperates in the first round, that must be satisfied for her to be able to enforce various linear payoff relationships (extortionate, generous, and equalizer) using a two-point autocratic strategy.

Using the reaction probabilities defined by
\begin{linenomath}
\begin{subequations}
\begin{align}
p^{X}\left(x\right) &= \frac{b\left(x\right) +\chi c\left(x\right) +\chi\kappa_{Y}-\kappa_{X}}{\lambda\omega_{X}\left(\chi +1\right)\!\left(b\left(K\right) +c\left(K\right)\!\right)} - \frac{1-\lambda}{\lambda} p_{0} ; \\
p^{Y}\left(y\right) &= p^{X}\left(y\right) ,
\end{align}
\end{subequations}
\end{linenomath}
player $X$ can enforce $\pi_{X}-\kappa_{X}=\chi\left(\pi_{Y}-\kappa_{Y}\right)$ provided $p_{0}$ satisfies
\begin{linenomath}
\begin{align}\label{sieq:extGenRandomInitial}
\max &\left\{\frac{b\left(K\right) +\chi c\left(K\right) + \chi\kappa_{Y}-\kappa_{X}}{\left(1-\lambda\right)\omega_{X}\!\left(\chi +1\right)\left( b\left(K\right) + c\left(K\right) \right)}-\frac{\lambda}{1-\lambda},\ 0\right\} 
\leqslant p_{0} \leqslant \min\left\{\frac{\chi\kappa_{Y}-\kappa_{X}}{\left(1-\lambda\right)\omega_{X}\!\left(\chi +1\right)\left(b\left(K\right) +c\left(K\right)\right)},\ 1\right\} .
\end{align}
\end{linenomath}
Similarly, player $X$ can ensure $\pi_{Y}=\gamma$ by using the reaction probabilities
\begin{linenomath}
\begin{subequations}
\begin{align}
p^{X}\left(x\right) &= \frac{c\left(x\right) + \gamma}{\lambda\,\omega_{X}\left( b\left(K\right) + c\left(K\right) \right)} - \frac{1-\lambda}{\lambda} p_{0} ; \\
p^{Y}\left(y\right) &= p^{X}\left(y\right)
\end{align}
\end{subequations}
\end{linenomath}
as long as $p_{0}$ falls within the range
\begin{linenomath}
\begin{align}\label{sieq:equalizerRandomInitial}
\max\left\{\frac{c\left(K\right) +\gamma}{\left(1-\lambda\right)\omega_{X}\left(b\left(K\right) +c\left(K\right)\right)} - \frac{\lambda}{1-\lambda} ,0\right\} \leqslant p_{0} \leqslant \min\left\{\frac{\gamma}{\left(1-\lambda\right)\omega_{X}\left(b\left(K\right) +c\left(K\right)\right)} ,1\right\} .
\end{align}
\end{linenomath}

\end{document}